\documentclass[11pt,3p,times]{elsarticle}

\usepackage{amsmath,amsfonts,amssymb}
\usepackage{nicematrix}
\usepackage{scalerel}
\usepackage{amsthm}
\usepackage{comment}
\usepackage{mathtools}
\usepackage{xcolor}
\usepackage{tikz}
\usepackage{soul}
\usepackage{cancel}
\usepackage{float}
\usepackage{bigints}

\usepackage{xpatch}
\usepackage{accents}
\usepackage{multicol}
\usepackage{physics}

\newcommand{\proofnameformat}{\itshape}
\xpatchcmd{\proof}{\itshape}{\proofnameformat}{}{}
\renewcommand{\proofnameformat}{\itshape\color{astral}}

\renewcommand{\ddot}[1]{%
  \accentset{\mbox{\large\bfseries .\hspace{0.5ex}.}}{#1}}
  
\usepackage{amsthm,thmtools,xcolor}
\usetikzlibrary{arrows.meta}

\declaretheoremstyle[
  headfont=\color{astral}\normalfont\bfseries,
]{colored}

\declaretheorem[
  style=colored,
  name=Theorem,
]{thm}

\declaretheorem[
  style=colored,
  name=Proposition,
]{proposition}

\declaretheorem[
  style=colored,
  name=Lemma,
]{lem}

\declaretheorem[
  style=colored,
  name=Remark,
]{remark}

\declaretheorem[
  style=colored,
  name=Definition,
]{definition}

\declaretheorem[
  style=colored,
  name=Cauchy-Peano Theorem,
  numbered=no
]{CP}

\declaretheorem[
  style=colored,
  name=Gronwall's inequality,
  numbered=no
]{GI}

\newcommand{\R}{{\mathbb{R}}}

\newcommand{\C}{{\mathbb{C}}}

\newcommand{\vertiii}[1]{{\left\vert\kern0.1ex \left\vert\kern0.1ex \left\vert #1 \right\vert\kern0.1ex\right\vert\kern0.1ex \right\vert}}

\usepackage{sectsty}

\definecolor{astral}{RGB}{46,116,181}
\subsectionfont{\color{astral}}
\sectionfont{\color{astral}}

\usepackage{amssymb}
\usepackage{amsthm}

\usepackage{fancyhdr}
\usepackage[figuresright]{rotating}

\usepackage{hyperref}
\hypersetup{
    colorlinks = false,
    linkbordercolor = {green}
}

\begin{document}

\begin{frontmatter}




\title{\color{astral}{\huge \textbf{Dynamical system for PageRank with a time-dependent memory}}}


\author{\color{black}{Julio Flores$^{1,2}$, Eva Primo$^{1,2,3}$, Daniel Rodríguez$^{1,2,3}$ and Miguel Romance$^{1, 2, 3}$}}

\address{$^{1)}$Departamento de Matemática Aplicada, Ciencia e Ingeniería de los Materiales y Tecnología Electrónica,
Universidad Rey Juan Carlos, 28933 Móstoles (Madrid), Spain\\
$^{2)}$Laboratory of Mathematical Computation on Complex Networks and their Applications, Universidad Rey Juan Carlos,
28933 Móstoles (Madrid), Spain\\
$^{3)}$Data, Complex networks and Cybersecurity Research Institute, Universidad Rey Juan Carlos, 28028 (Madrid), Spain}

\begin{abstract}

Inspired by the dynamical PageRank framework of Gleich and Rossi \cite{GleichRossi}, we introduce a continuous-time PageRank model in which the personalization vector evolves as a weighted average of its past values, with the weights determined by a memory function. The resulting dynamics are formulated as an initial value problem for an integro-differential equation, where the initial condition is a probability vector. We investigate how the choice of memory function influences the long-time behavior of the PageRank vector. In particular, for strongly connected networks $\mathcal{G}$, we prove that broad classes of memory functions lead to convergence toward a stationary state that is independent of the initial condition. In contrast, when the memory function is exponential-oscillatory, $\omega(t)=e^{at}\cos(bt)$ for $t\geq0$ with $a,b>0$, we show that the PageRank dynamics exhibit asymptotically periodic behavior, revealing that oscillatory memory can fundamentally alter the qualitative evolution of the ranking process. To establish these results, we first prove the existence and uniqueness of solutions using standard results from the theory of integro-differential equations and show that the solution remains a probability vector for all times, thereby preserving the essential properties of the PageRank model.\\

\hspace*{-5mm}{\bf\color{astral}{Keywords:}} PageRank, irreducible matrix, Perron-Frobenius theory, strongly connected network, dynamical system, perturbation theory, asymptotic behavior\\
 
\hspace*{-5mm}{\bf\color{astral}{AMS subject classifications:}} 15A60, 45J05, 45M05, 45M10, 45M15, 47J26
 \end{abstract}




\end{frontmatter}


\section{Introduction}
\label{Section:Introduction}

PageRank is one of the most influential centrality measures in Network Science, providing a mathematically rigorous framework for quantifying the relative importance of nodes in directed networks. Originally introduced by Page and Brin as the ranking mechanism underlying Google's search engine \cite{BrinPage,PageBrin}, the algorithm rapidly became a cornerstone of graph-based data analysis. Beyond its original application to the World Wide Web, PageRank has evolved into a fundamental tool for studying information flow, influence, and structural relevance across a wide range of complex systems, owing to its solid probabilistic interpretation and computational efficiency.

The versatility of PageRank has fostered its adoption in numerous scientific disciplines. In social networks, it is widely employed to identify influential users and characterize patterns of information diffusion \cite{Riquelme}. In systems biology, PageRank-based methods have been successfully applied to protein--protein interaction networks and metabolic pathways \cite{Ivan}. Financial networks exploit PageRank to assess systemic risk by identifying institutions whose failure may trigger cascading contagion phenomena \cite{Yun}. Additional applications include the analysis of structural connectivity in neuroscience \cite{Williamson}, the evaluation of scientific impact beyond conventional citation metrics \cite{Senanayake}, and the study of transportation and mobility networks \cite{Criado}. These diverse applications illustrate the remarkable generality of the PageRank paradigm and its relevance as a universal measure of network centrality.

From a mathematical perspective, PageRank can be interpreted as the stationary probability distribution of a Markov chain describing the movement of a random surfer on a directed graph $\mathcal{G}=(V,E)$. At each step, the walker follows one of the outgoing edges of the current node with a prescribed probability or, alternatively, performs a random jump to another node according to a given probability distribution, commonly referred to as the personalization vector $\mathbf{v}$. The balance between these two mechanisms is controlled by the damping factor, which guarantees ergodicity even in the presence of dangling nodes or disconnected components (see, for example, \cite{Gleich,LM2006,PageBrin}). Consequently, the computation of PageRank fundamentally depends on two parameters: the damping factor governing the navigation dynamics and the personalization vector defining the teleportation process.

The mathematical properties of PageRank have been extensively investigated during the last two decades. Numerous studies have established rigorous results concerning existence, uniqueness, convergence of iterative algorithms, perturbation analysis, and the sensitivity of the ranking with respect to both the damping factor and the personalization vector \cite{Bianchini,Boldi,Gleich,LM2004,LM2006}. These theoretical developments have significantly deepened our understanding of PageRank and have contributed to its widespread adoption in both theoretical and applied settings.

More recently, attention has shifted from the classical static formulation toward models incorporating feedback mechanisms and adaptive dynamics. In particular, the authors of \cite{AFPRR} investigated the dependence of PageRank on the personalization vector through a feedback-PageRank framework, providing a complete characterization of the existence and uniqueness of PageRank fixed points in terms of the strongly connected components of the underlying graph. Furthermore, they showed that these fixed points can be computed exactly through the classical Power Method by exploiting the left-hand Perron eigenvectors associated with each strongly connected component. The rationale underlying these feedback mechanisms originates from the original interpretation of PageRank \cite{PageBrin}, in which the personalization vector $\mathbf{v}$ models the probability distribution governing users' navigation when they choose not to follow an outgoing hyperlink. From this perspective, it is natural to extend the classical formulation by allowing the personalization vector to evolve over time, as proposed in \cite{GleichRossi}. Moreover, a further and more realistic generalization consists in assuming that its evolution is influenced by the past dynamics of the PageRank vector itself, thereby introducing a feedback mechanism with memory \cite{AFPRR}.

The present work builds upon this line of research by introducing a substantially more general dynamical framework in which the evolution of the PageRank vector depends not only on its current state but also on its past history through a memory mechanism. Unlike the classical PageRank formulation, which computes a stationary distribution of a Markovian random walk, and existing feedback-based models in which the personalization vector depends only on the current network state \cite{LM2006}, the proposed approach introduces a time-dependent memory mechanism that allows the present dynamics to be influenced by their entire past evolution. This formulation naturally captures non-Markovian effects that arise in many real-world networked systems, where user behavior and information propagation are inherently history-dependent.

From a mathematical perspective, the proposed model is formulated as a linear integro-differential dynamical system whose memory term is described by a general time-dependent kernel. This setting considerably extends previous feedback-PageRank formulations while preserving the probabilistic interpretation of the ranking process. The analysis combines techniques from the theory of integro-differential equations, dynamical systems, perturbation theory and Perron--Frobenius theory to establish the fundamental properties of the proposed model. In particular, we first prove that the associated initial value problem is well posed by establishing the existence and uniqueness of global solutions and showing that the PageRank vector remains in the probability simplex for all times, thereby preserving the essential stochastic structure of the classical PageRank algorithm.

In addition to this, the second part of the paper investigates the asymptotic behavior of the resulting dynamical system on strongly connected directed networks. We show that the long-time dynamics depend critically on the qualitative properties of the memory kernel. For broad classes of kernels, including exponentially decaying memory, the PageRank dynamics converge asymptotically to a unique stationary ranking, extending the classical PageRank equilibrium to a memory-dependent setting. In contrast, when the memory kernel exhibits exponentially damped oscillatory behavior, the dynamics no longer converge to a fixed point but instead evolve toward asymptotically periodic solutions, revealing that memory may fundamentally alter the qualitative behavior of network centrality. These results provide, to the best of our knowledge, the first rigorous analytical characterization of memory-induced dynamical transitions in a feedback PageRank model.

Beyond their theoretical interest, these findings contribute to the broader study of adaptive and temporal complex networks by providing a mathematically rigorous framework for incorporating historical dependence into centrality measures. The proposed model extends the applicability of PageRank to systems in which past interactions continue to influence future dynamics, including information diffusion, social influence, recommendation systems, and adaptive communication networks. More generally, the results establish a bridge between classical network centrality, non-Markovian stochastic processes, and memory-driven dynamical systems, thereby opening new directions for the analysis of evolving complex networks.

The remainder of the paper is organized as follows. In Section \ref{Section:Notation} we fix the notation and some preliminary definitions related to the PageRank algorithm. In Section \ref{Section:Main results}, we show the existence and uniqueness of solutions to the initial value problem (IVP for short) for the dynamical system governing the PageRank vector with time-dependent memory (see Theorem \ref{thm:existence-uniqueness}), defined by the integro-differential equation
\begin{equation}\label{eq:IVP-intro}
\dot{\mathbf{x}}^\top(t)=\mathbf{x}^\top(t)(\lambda P_A-I_N)+(1-\lambda)\left[\frac{1}{\displaystyle\int_0^t \omega(u)\,du}\int_0^t \omega(t-u)\,\mathbf{x}^\top (u)\,du\right],\quad \mathbf{x}(0)=\mathbf{x}_0,
\end{equation}
where $\mathbf{x}_0 \in \mathbb{R}^{N \times 1}$ is a probability vector (i.e., $\mathbf{x}_0 \geq 0$ and $\mathbf{x}_0^\top \mathbf{e} = 1$ with  $\mathbf{e} = (1,\dots,1)^\top$), $P_A$ is the row-normalized adjacency matrix of the directed graph $\mathcal{G}$, $\lambda \in (0,1)$ is the damping factor, and $\omega : [0,\infty) \to [0,\infty)$ is a locally integrable weight function. Moreover, we prove that the solution $\mathbf{x}(t)\in\mathbb{R}^{N\times 1}$ of the dynamical system in equation (\ref{eq:IVP-intro}) remains a probability vector for all $t\geq 0$, that is, $\mathbf{x}(t)\geq 0$ (see Theorem \ref{thm:positivity}) and $\mathbf{x}^\top(t)\,\mathbf{e}=1$ (see Theorem \ref{thm:identity-to-one}) for all $t\geq0$. For the sake of clarity, these two facts are stated and proved separately.

Later, in Section \ref{Section:Asymptotic}, we investigate the asymptotic behavior of the solution $\mathbf{x}(t)$ as $t\to\infty$ on a strongly connected digraph $\mathcal{G}$ for several classes of memory weight functions, damping factors $\lambda$, and initial conditions $\mathbf{x}_0\in\mathbb{R}^{N\times1}$. We begin in Subsection \ref{subsect:asymptotic-delta} by considering the Dirac delta memory function $\omega(t)=\delta_0(t)$ for $t\geq 0$. In this setting, Theorem \ref{thm:delta-case} establishes that the solution $\mathbf{x}(t)$ of (\ref{eq:IVP-intro}) converges to the left-hand Perron vector of $P_A$, whose definition is given in Section \ref{Section:Main results}. Next, in Subsection \ref{subsect:asymptotic-exponential}, we consider the exponentially decaying memory function $\omega(t)=e^{-at}$, for $t\geq0$, where $a\in\R$ is a positive parameter. As in the previous case, Theorem \ref{thm:exponential-case} shows that the solution $\mathbf{x}(t)$ of the IVP in equation (\ref{eq:IVP-intro}) converges to the left-hand Perron vector of row-normalized matrix $P_A$. It is worth mentioning that this convergence is independent of the damping factor $\lambda\in(0,1)$ and of the probability vector $\mathbf{x}_0\in\mathbb{R}^{N\times1}$ used as the initial condition. Finally, in Subsection \ref{subsect:asymptotic-oscillatory}, we consider a exponential oscillatory memory function  $\omega(t)=e^{at}\cos^2(bt)$, for $t\geq0$, where $a,b\in\R$ are positive parameters. In contrast to the cases studied in Subsections \ref{subsect:asymptotic-delta} and \ref{subsect:asymptotic-exponential}, Theorem \ref{thm:oscillatory-case} shows that the solution $\mathbf{x}(t)$ of the IVP in equation \ref{eq:IVP-intro} is asymptotically periodic, with no apparent relation to the left-hand Perron vector of $P_A$. More precisely, $\mathbf{x}(t)$ converges to a $\pi/b$-periodic function as $t\to\infty$ and this limiting periodic function depends strongly on the damping factor $\lambda\in(0,1)$ and the initial condition $\mathbf{x}_0\in\R^{N\times1}$ considered.

\section{Notation and preliminary definitions}
\label{Section:Notation}


We recall some standard notation that will be used throughout the paper. Vectors of $\mathbb{R}^{N\times 1}$ will be denoted by column matrices and we will use the superscript $T$ to indicate matrix transposition. The vector of $\mathbb{R}^{N\times 1}$ with all its components equal to 1 will be denoted by $\mathbf{e}$, that is, $\mathbf{e} = (1,\dots,1)^\top$ . A matrix $A=(a_{ij})_{i,j}$ will be called \em non-negative \em (respectively \em positive \em) if all its entries satisfy $a_{ij}\ge 0$ (respectively $a_{ij}>0$). The same applies to vectors when being column matrices.\\

Let $\mathcal{G}=(V,E)$ be a directed graph  with node set $V=\{1,2,\dots,N\}$ for some $N\in\mathbb{N}$. The edge set $E$ consists of ordered pairs $(i,j)$, representing a directed link from node $i$ to node $j$. Throughout this paper, we use the terms directed graph and digraph interchangeably. The \textit{adjacency matrix} of $\mathcal{G}$ is defined as the $N\times N$ matrix $A=(a_{ij})$, given by
\begin{equation*}
a_{ij}=
\begin{cases}
    1, & \text{if } (i,j) \in E, \\
    0, & \text{otherwise.}
\end{cases}
\end{equation*}

A link $(i,j)$ is said to be an \textit{outlink} for node $i$ and an \textit{inlink} for node $j$. We denote $k_{out}(i)$ the \textit{outdegree} of node $i$, i.e, the number of outlinks of a node $i$. Notice that $k_{out} (i) = \sum_{k}a_{ik}$. A node will be referred to as a \em source \em if it has no inlinks. Also, the graph $\mathcal{G}$ may have \textit{dangling nodes}, which are nodes $i\in V$ with zero outdegree. Dangling nodes are characterized by a vector $\mathbf{d}=(d_1,\cdots,d_N)^\top\in\mathbb{R}^{N\times 1}$ with components $d_i$ defined by
$$d_{i}=
\begin{cases}
    1,\text{ if $i$ is a dangling node of $\mathcal{G}$,}\\
    0,\text{ otherwise}.
\end{cases}$$
Without loss of generality,  the directed graph $\mathcal{G}$ can be weighted with $a_{ij}\geq0$. We remark that the results in the paper remain valid if $\mathcal{G}$ is weighted.\\

At this point, let $P_A = (p_{ij})_{i,j}\in\mathbb{R}^{N\times N}$ be the row-stochastic matrix associated to directed graph $\mathcal{G}$ defined in the following way:
\begin{enumerate}
    \item If $i$ is a dangling node, then $p_{ij}=0$ for all $j=1,\dots,N$.
    \item Otherwise, $p_{ij}=a_{ij}/k_{out}(i)=a_{ij}/\sum_{k}a_{ik}$.
\end{enumerate}

Roughly speaking, each coefficient $p_{ij}$ is the probability of moving from the node $i$ to the node $j$. The matrix $P_A$ will be referred to as the \em row-normalization \em of $A$.  Note that if we take
\begin{equation}\label{eq:diagonal}
D=\left(
\begin{array}{cccc}
k_{out}(1)& 0 & \cdots & 0\\
0 & k_{out}(2) & \cdots & 0\\
\vdots & \cdots & \ddots & \vdots\\
0 & \cdots & \cdots & k_{out}(N)
\end{array}
\right),
\end{equation}
then the matrix $P_A$ is given by $P_A=D^{-1}A$.\\

One of the most remarkable features of the PageRank algorithm is that it contemplates the possibility to travel from one node to any other node in the graph. This teleportation probability is given by the \textit{personalization vector} $\textbf{v}\in\R^{N\times1}$ which, in fact, is a probability distribution vector. In addition, in the presence of dangling nodes, a probability vector $\mathbf{u}\in\mathbb{R}^{N\times1}$ should be considered to provide an extra probability of jumping from these nodes. Finally, let $\lambda\in(0,1)$ be a decision parameter, or \em damping \em factor, which determines the likelihood to move through the graph by using a path in  $\mathcal{G}$ or  by randomly jumping instead, according to  the personalization vector. For more information, we refer the reader to \cite{AFPR}.\\

Formally, let $G=G(\lambda, \mathbf{u},\mathbf{v})$ be the \em Google matrix\em, with $\lambda\in(0,1)$ defined as
\begin{equation}\label{eq:google-matrix1}
G=\lambda (P_A+\mathbf{du}^\top)+(1-\lambda)\mathbf{ev}^\top\in\mathbb{R}^{N\times N}.
\end{equation} 

A straightforward observation is that the Google matrix $G$ is row-stochastic, i.e., $G\mathbf{e} = \mathbf{e}$ where $\mathbf{e}=(1,\dots,1)^\top$. Recall that the personalization vector $\mathbf{v}\in\mathbb{R}^{N\times 1}$ is such that all its entries are positive (i.e. $\mathbf{v}>0$) and satisfies $\mathbf{v}^\top\mathbf{e} = 1$. In a similar way, the vector $\mathbf{u}\in\mathbb{R}^{N\times1}$ is such that $\mathbf{u}>0$ and $\mathbf{u}^\top\mathbf{e} = 1$.\\

Notice that the Google matrix $G$ is positive and thus, by the Perron-Frobenius Theorem (see \cite[Section 8.3]{Meyer}) there is a unique vector  $\pi= \pi(\lambda;\mathbf{u};\mathbf{v})>0$ satisfying $\pi^\top\textbf{e}=1$  and $G^\top\pi=\pi$, or equivalently $\pi^\top G=\pi^\top$. This vector is the \textit{PageRank vector} $\pi$ of $G$. \\

The existence of dangling nodes will not affect the results presented for the matrix $G$ (see the end of paper \cite{AFPR} for more details). In this sense, for the sake of simplicity, the graph $\mathcal{G}$ will be assumed to have no dangling nodes.\\

From the definition of the Google matrix in equation (\ref{eq:google-matrix1}) and the properties of the PageRank vector $\pi\in\R^{N\times1}$, a straightforward calculation shows that $\pi$ satisfies the equation
$$\pi^\top=\lambda \pi ^\top P_A+(1-\lambda)\mathbf{v}^\top.$$

Notice that $\pi^\top(I_N-\lambda P_A)=(1-\lambda)\mathbf{v}^\top$, where $I_N\in\R^{N\times N}$ is the identity square matrix of order $N$. Therefore, the PageRank vector $\pi$ can be explicitly calculated as
\begin{equation}\label{eq:pi-vector}
\pi^\top=(1-\lambda)\mathbf{v}^\top(I_N-\lambda P_A)^{-1}=\mathbf{v}^\top X(\lambda),
\end{equation}
where $X(\lambda)=(1-\lambda)R(\lambda)$ and $R(\lambda)=(I_N-\lambda P_A)^{-1}$ is the resolvent of $P_A$ defined for all suitable damping factor $\lambda\in (0,1)$ (see \cite{Boldi} for more information on this topic). At this stage, some authors proposed a feedback mechanism whereby the PageRank vector obtained at each iteration is employed as the new personalization vector in the subsequent step (see \cite{AFPRR} for more details). This construction defines a recursive sequence of vectors $\{\mathbf{x}_k\}_{k \ge 0} \subset \mathbb{R}^{N \times 1}$ given by
\begin{equation}
\mathbf{x}_k^\top = \mathbf{x}_{k-1}^\top X(\lambda) = \mathbf{x}_0^\top X^k(\lambda), \quad k \geq 1,
\end{equation}
where $\mathbf{x}_0 \in \mathbb{R}^{N \times 1}$ is an initial probability vector (i.e., $\mathbf{x}_0\geq0$ and $\mathbf{x}_0^\top\mathbf{e}=1$). The proposed {feedback-PageRank mechanism} provides a complete characterization of the existence and uniqueness of PageRank fixed points, via the classical power method, in terms of left-hand Perron vector associated with each strongly connected component of the directed graph $\mathcal{G}$.

The PageRank vector, characterized as the fixed point of the Google matrix $G$, can be obtained as the limit of the recursive sequence $\{\mathbf{x}_k\}_{k \ge 0} \subset \mathbb{R}^{N \times 1}$ defined by
\begin{equation}
\mathbf{x}_k^\top = \mathbf{x}_{k-1}^\top G
= \lambda \mathbf{x}_{k-1}^\top P_A + (1-\lambda)\mathbf{v}^\top,
\quad k \geq 1,
\end{equation}
where it is assumed that $\mathbf{x}_k^\top \mathbf{e} = 1$ for all $k \geq 0$. In this formulation, and in the spirit of \cite{AFPR}, if the personalization vector is updated at each iteration so as to reflect a weighted average of their preceding values, then the resulting process gives rise to the following iterative scheme 
\begin{equation*}\label{eq:average-previous}
\mathbf{x}^\top_{k+1}=\lambda \mathbf{x}^\top_{k}P_A+(1-\lambda)\frac{1}{\sum_{i=0}^k \omega_i}\sum_{i=0}^k \omega_{k-i}\,\mathbf{x}^\top_{i},
\end{equation*}

where $\{\omega_k\}_{k\geq0}\subset\R$ is a sequence of positive real numbers. At this point, in order to describe how the sequence $\{\mathbf{x}_k\}_{k\geq0}$ evolves over time, we subtract $\mathbf{x}_k$ from the previous equation to obtain the relation 
\begin{equation}\label{eq:increment-delta}
\Delta\mathbf{x}^\top_k=\mathbf{x}^\top_{k+1}-\mathbf{x}^\top_{k}=\left[\lambda \mathbf{x}^\top_{k}P_A+(1-\lambda)\frac{1}{\sum_{i=0}^k \omega_i}\sum_{i=0}^k \omega_{k-i}\,\mathbf{x}^\top_{i}\right]-\mathbf{x}^\top_{k}
=\mathbf{x}^\top_{k}(\lambda P_A-I_N)+(1-\lambda)\frac{1}{\sum_{i=0}^k \omega_i}\sum_{i=0}^k \omega_{k-i}\,\mathbf{x}^\top_{i}
\end{equation}
Therefore, changes in the values of the sequence $\{\mathbf{x}_k\}_{k\geq0}$ at each node are closely related to the increment $\mathbf{x}^\top_{k}(\lambda P_A-I_N)$ and to the weighted average of their preceding values. In fact, equation (\ref{eq:increment-delta}) can be interpreted as a continuous-time dynamical system given by

\begin{equation}\label{eq:probability-vector}
\dot{\mathbf{x}}^\top(t)=\frac{d}{dt}\left(\mathbf{x}^\top(t)\right)=\mathbf{x}^\top(t)(\lambda P_A-I_N)+(1-\lambda)\left[\frac{1}{\displaystyle\int_0^t \omega(u)\,du}\int_0^t \omega(t-u)\,\mathbf{x}^\top (u)\,du\right].
\end{equation}
Notice that the idea of proposing a dynamical system for the PageRank vector was studied by Gleich and Rossi in \cite{GleichRossi}, where the authors reformulated the classical PageRank algorithm in terms of changes in the PageRank values for each page. This reformulation, also based on Richardson's iteration for the PageRank system $\mathbf{x}^\top(I_N-\lambda P_A)=(1-\lambda)\mathbf{v}^\top$, enabled them to express the PageRank as the dynamical system 
\begin{equation}\label{eq:dynamical-gleich}
\dot{\mathbf{x}}^\top(t)=\frac{d}{dt}\left(\mathbf{x}^\top(t)\right)=(1-\lambda)\mathbf{v}^\top(t)-\mathbf{x}^\top(t)(I_N-\lambda P_A),
\end{equation}
where a time-dependent personalization vector $\mathbf{v}(t)$ was considered.

This approach captures the evolution of the PageRank vector, providing a natural mechanism for integrating temporal changes. Unlike the static one, this dynamical approach produces time-varying importance scores that reflect fluctuating external interest in specific nodes. Note that in this approach the personalization vector is a prefixed function (for instance, $\mathbf{v}^\top(t)=\frac{1}{k}\sum_{j=1}^k\mathbf{v}_j^\top\left[\cos(t+(j-1)\frac{2\pi}{k})+1\right]$ for some probability vectors $\mathbf{v}_1,\dots,\mathbf{v}_k\in\R^{N\times1}$) independent of the intrinsic dynamical behavior of the system, while in our model the personalization vector explicitly depends on the previous values of the vector $\mathbf{x}(t)$. 


\section{Main results}
\label{Section:Main results}

In this section we show some properties concerning the solution for the Initial Value Problem (IVP) given by 
\begin{equation}\label{eq:IVP}
\dot{\mathbf{x}}^\top(t)=\mathbf{x}^\top(t)(\lambda P_A-I_N)+(1-\lambda)\left[\frac{1}{\displaystyle\int_0^t \omega(u)\,du}\int_0^t \omega(t-u)\,\mathbf{x}^\top (u)\,du\right],\quad \mathbf{x}(0)=\mathbf{x}_0,
\end{equation}
where $\mathbf{x}_0\in\R^{N\times1}$ is a probability vector (i.e., $\mathbf{x}_0\geq0$ and $\mathbf{x}_0^\top\mathbf{e}=1$), $P_A$ is the row-normalization of the adjacency matrix $A$ associated to directed graph $\mathcal{G}$, $\lambda$ is the damping factor (with $0<\lambda<1$) and the weight function $\omega$ is  non-negative and locally integrable on the interval $[0,\infty)$, that is, its Lebesgue integral is finite on all compact set of $[0,\infty)$.

For $p\geq1$, the $p$-norm of a vector $\mathbf{x}\in\R^{N\times1}$ with $\mathbf{x}^\top=(x_1,\dots,x_N)$ is defined as 
\begin{equation*}
\Vert\, \mathbf{x}^\top\,\Vert_p=\left(\sum_{k=1}^N\vert\, x_k\,\vert^p\right)^{1/p}. 
\end{equation*} 
Note that, for a non-negative vector $\mathbf{x}\in\R^{N\times1}$, the identity $\mathbf{x}^\top\mathbf{e}=1$ can be expressed in terms of the 1-norm as $\Vert \mathbf{x^\top}\Vert_1 = 1$. Throughout the paper we mainly use the $p$-norm with $p=1$.\\

Before proceeding, we recall the several-variable version of the Cauchy-Peano Theorem concerning the existence of local solutions to initial value problems, which plays a key role in establishing the well-posedness of the IVP in equation (\ref{eq:IVP}). We now present this result as it applies to our context. From now on, and throughout the paper, we will use the notation $\dot{\mathbf{x}}(t)$ for the derivative of $\mathbf{x}(t)$ with respect to time to make presentation cleaner.\\

\begin{CP}[Theorem 8.27 in \cite{Kelley-Peterson}]\label{thm:CP}
Let $a,b$ be positive real numbers. Assume $t_0\in \R$, $\mathbf{x}_0\in \R^{N\times1}$ is a probability vector (i.e., $\mathbf{x}_0\geq0$ and $\mathbf{x}_0^\top\mathbf{e}=1$) and let $f:\R\times\R^{N\times1}\to\R^{N\times1}$ be a continuous $(N+1)$-dimensional vector function on the rectangle
$$Q=\left\{\,(t,\mathbf{x})\in \R\times\R^{N\times1}\,:\,\vert t-t_0\vert\leq a \,,\,{\Vert \,\mathbf{x}^\top-\mathbf{x}_0^\top\,\Vert_1\leq b}\,\right\}.$$
Then the Initial Value Problem (IVP)
$$\dot{\mathbf{x}}^\top(t)=f(t,\,\mathbf{x}^\top(t)),\quad \mathbf{x}(t_0)=\mathbf{x}_0,$$
has a solution $\mathbf{x}$ on the interval $[t_0-\alpha,t_0+\alpha]$ with $\Vert\,\mathbf{x}(t)-\mathbf{x}_0\,\Vert_1\leq b$, where 
$$\alpha=\min\left\{a,\,\frac{b}{M}\right\}\quad \text{ and }\quad  M=\max\left\{\,\Vert\,f(t,\,\mathbf{x}^\top)\,\Vert_1\,:\,(t,\,\mathbf{x})\in Q\,\right\}.$$
\end{CP}

\begin{remark}\label{rmk:operator-norm}
On  the space $\mathbb{R}^{N\times N}$ of square matrices of order $N$, we will consider the well-known (max-column) matrix norm $$\||\,Q\,|\|_1=\max_{j=1,\dots,N}\sum_i|a_{ij}|$$  and  the  (max-row) matrix norm $$\||\,Q\,|\|_\infty=\max_{i=1,\dots,N}\sum_j|a_{ij}|$$ (see \cite{Horn-Johnson},  Example 5.6.4 and Example 5.6.5). Clearly  $\||\,Q\,\||_1=\||\,Q^\top\,\||_\infty$. Additionally $\||\cdot\||_1$ coincides with the operator norm on $\mathbb{R}^{N\times N}$  induced by the $l_1$-norm on $\mathbb{R}^N$, that is,   $\|\,Q\mathbf{x}\,\|_1\le \||\,Q\,\||_1\|\,\mathbf{x}\,\|_1$ (see Example 5.6.4 in \cite{Horn-Johnson}). Thus, for every $\mathbf{x}\in\R^{N\times 1}$ we have
$$\|\,\mathbf{x}^\top Q\,\|_1=\|\,(\mathbf{x}^\top Q)^\top\,\|_1=\|\,Q^\top\,\mathbf{x}\,\|_1\le\||\,Q^\top\,\||_1\|\,\mathbf{x}\,\|_1=\||\,Q\,\||_\infty\|\,\mathbf{x}\,\|_1=\||\,Q\,\||_\infty\|\,\mathbf{x}^\top\,\|_1.$$
In particular, if $Q$ is a row-stochastic matrix we get $\|\,\mathbf{x}^\top Q\,\|_1\le \|\,\mathbf{x}^\top\,\|_1.$
\end{remark}

We also recall Gronwall’s inequality in its integral form as it is a key ingredient in the proof of unicity. More precisely,

\begin{GI}[Chapter 2 in \cite{Burton}]\label{thm:Gronwall}
Let $f,g:[0,T]\to[0,\infty)$ be continuous functions and $\alpha\ge 0$ real. If 
$$f(t)\leq \alpha+\int_0^tg(s)f(s)\,ds, \quad 0\leq t\leq T,$$
then 
$$f(t)\leq \alpha\exp\left(\int_0^tg(s)\,ds\right), \quad 0\leq t\leq T.$$
\end{GI}
At this point, the Cauchy–Peano Theorem and Gronwall’s inequality mentioned above, will allow us to state the main theorem of this paper concerning the existence and uniqueness of solutions to the IVP in equation (\ref{eq:IVP}). More precisely,

\begin{thm}\label{thm:existence-uniqueness}
Let $A$ be the adjacency matrix of a directed graph $\mathcal{G}$ and let $P_A$ be its row-normalization
described in Section \ref{Section:Introduction}. Let $\mathbf{x}_0\in\R^{N\times1}$ be a probability vector (i.e., $\mathbf{x}_0\geq0$ and $\mathbf{x}_0^\top\mathbf{e}=1$) and let $\lambda\in(0,1)$ be a fixed damping factor. If $\omega:[0,\infty)\to[0,\infty)$ is a locally integrable function, then the Initial Value Problem (IVP) given by
\begin{equation}\label{eq:IVP-vectorial}
\dot{\mathbf{x}}^\top(t)=\mathbf{x}^\top(t)(\lambda P_A-I_N)+(1-\lambda)\left[\frac{1}{\displaystyle\int_0^t \omega(u)\,du}\int_0^t \omega(t-u)\,\mathbf{x}^\top (u)\,du\right],\quad \mathbf{x}(0)=\mathbf{x}_0,
\end{equation}
has a unique solution $\mathbf{x}$ on the interval $[0,\infty)$.
\end{thm}

\begin{proof}
Let $a,b$ be positive real numbers and let us consider the rectangle
$$Q=\left\{(t,\mathbf{x})\in \R\times\R^{N\times1}\,:\,0\leq t\leq a \,,\,{\Vert\, \mathbf{x}^\top-\mathbf{x}_0^\top\,\Vert_1\leq b}\,\right\}.$$
where $\mathbf{x}_0\in\R^{N\times1}$ is a probability vector, i.e., $\mathbf{x}_0\geq0$ and $\mathbf{x}_0^\top\mathbf{e}=1$. Now, let consider the IVP given by 
$$\dot{\mathbf{x}}^\top(t)=f(t,\,\mathbf{x}^\top(t)), \quad \mathbf{x}(0)=\mathbf{x}_0,$$
where 
\begin{equation*}
f(t, \mathbf{x}^\top(t))=\mathbf{x}^\top(t)(\lambda P_A-I)+(1-\lambda)\left[\frac{1}{\displaystyle\int_0^t \omega(u)\,du}\int_0^t \omega(t-u)\,\mathbf{x}^\top (u)\,du\right].
\end{equation*}
It is clear that $f$ is a continuous function on the rectangle $Q$. Indeed, since $\omega$ is Lebesgue-integrable 
the function $t\rightarrow \int_0^t \omega(t-u)\,\mathbf{x}(u)\,du$ is continuous (in fact it is absolutely continuous \cite[Prop. 4.4.6]{Cohn}). Notice also that for $t=0$ we define $\left(\frac{1}{\displaystyle\int_0^t \omega(u)\,du}\int_0^t \omega(t-u)\,\mathbf{x}(u)\,du\right)=0$.\\

As an application of the Cauchy-Peano Theorem mentioned above, the IVP described in equation (\ref{eq:IVP-vectorial}) has a solution $\mathbf{x}$ on the interval $[0,T]$, for some $T>0$.\\

Now we proceed to show that the solution of the IVP in the vector case is unique. To this aim, suppose $\mathbf{x}_1(t)$ and $\mathbf{x}_2(t)$ are two different solutions of the integro-differential equation defined on the interval $[0,T]$ satisfying the same initial condition, that is, $\mathbf{x}_1(0)=\mathbf{x}_2(0)= \mathbf{x}_0$. Let us define the difference vector-valued function $\mathbf{z}(t) = \mathbf{x}_1(t) - \mathbf{x}_2(t)$ for $t\in[0,T]$. We need to prove that $\mathbf{z}(t)=\mathbf{0}_{N\times1}\in\R^{N\times 1}$ for all $t\in[0,T]$. Now, since the integro-differential equation in (\ref{eq:IVP-vectorial}) is a linear system, a straightforward observation is that the difference function $\mathbf{z}(t)$ satisfies the homogeneous equation
$$\dot{\mathbf{z}}^\top(t) = \mathbf{z}^\top(t)(\lambda P_A - I_N) + (1-\lambda)\left[\frac{1}{\displaystyle\int_0^t \omega(u)\,du}\int_0^t \omega(t-u)\,\mathbf{z}^\top(u)\,du\right]$$
with the initial condition $\mathbf{z}(0)=\mathbf{0}_{N\times1}$. Now, we apply the so-called variation of parameters method (see \cite[Section 4, Chapter 3]{Coddington-Levinson}). In order to solve the vector-valued integro-differential equation, we multiply by the matrix-integrating factor $e^{-( \lambda P_A-I_N)t}$. Since the initial condition is $\mathbf{z}(0) = \mathbf{0}_{N\times1}$, integrating both sides from $0$ to $t$ gives 
$$\mathbf{z}^\top(t) e^{-( \lambda P_A-I_N)t} = (1-\lambda)\int_0^t \left[\frac{1}{\displaystyle\int_0^{s} \omega(u)\,du}\int_0^s \omega(s-u)\,\mathbf{z}^\top(u)\,du \right]\,e^{-(\lambda P_A - I_N)s}\,ds.$$
Now, since the matrix $P_A$ commutes with the identity $I_N$, the inverse of the integrating factor exists and is, in fact, the matrix $e^{( \lambda P_A-I_N)t}$. Moreover, observe that $e^{( \lambda P_A-I_N)t}e^{-( \lambda P_A-I_N)s}=e^{( \lambda P_A-I_N)(t-s)}$ since $P_A$ commutes with itself. Therefore, multiplying both sides by the matrix $e^{( \lambda P_A-I_N)t}$ we obtain the equation
$$\mathbf{z}^\top(t) = (1-\lambda)\int_0^t \left[\frac{1}{\displaystyle\int_0^{s} \omega(u)\,du}\int_0^s \omega(s-u)\,\mathbf{z}^\top(u)\,du \right]\,e^{(\lambda P_A - I_N)(t-s)}\,ds.$$
At this point, we take the $1$-norm norm of both sides in the previous equation. Since $\omega$ is a non-negative function, we may move the norm inside the integrals and, as an application of Remark \ref{rmk:operator-norm}, obtain the scalar inequality
\begin{equation}\label{eq:1-norm-inequailty}
\Vert\,\mathbf{z}^\top(t)\,\Vert_1\leq(1-\lambda)\int_0^t \vertiii{\,e^{(\lambda P_A - I_N)(t-s)}\,}_{\infty }\left[\frac{1}{\displaystyle\int_0^{s} \omega(u)\,du}\int_0^s \omega(s-u)\,\,\Vert\,\mathbf{z}^\top(u)\,\Vert_1\,du \right]\,ds.    
\end{equation}
The next step is to provide a bound for the $\vertiii{\,\cdot\,}_{\infty}$-norm of the exponential matrix $e^{(\lambda P_A - I_N)(t-s)}$. Observe that, since the matrix $P_A$ commutes with the identity $I_N$, we can write 
$$e^{(\lambda P_A - I_N)(t-s)} = e^{-I_N(t-s)} e^{\lambda P_A(t-s)} = e^{-(t-s)} e^{\lambda P_A (t-s)}.$$
Notice that $\||P_A\||_\infty=1$ since $P_A$ is row-stochastic; thus  the Taylor expansion of the matrix-exponential gives  
$$\vertiii{\,e^{\lambda P_A (t-s)}\,}_{\infty} = \vertiii{\sum_{k=0}^\infty \,\frac{(\lambda P_A (t-s))^k}{k!}\,}_{\infty}\leq \sum_{k=0}^\infty \frac{\vertiii{\,\lambda P_A(t-s)\,}_{\infty}^k}{k!}=\sum_{k=0}^\infty \frac{\lambda^k (t-s)^k}{k!}\vertiii{\, P_A\,}_{\infty}^k=\sum_{k=0}^\infty \frac{\lambda^k (t-s)^k}{k!}=e^{\lambda (t-s)}.$$
Therefore $\vertiii{\,e^{(\lambda P_A - I_N)(t-s)}\,}_{\infty} \leq e^{-(t-s)} e^{\lambda (t-s)} = e^{-(1-\lambda)(t-s)}$. Substituting this estimation in equation (\ref{eq:1-norm-inequailty}) gives the inequality
$$\Vert\,\mathbf{z}^\top(t)\,\Vert_1 \leq(1-\lambda)\int_0^t e^{-(1-\lambda)(t-s)} \left[\frac{1}{\displaystyle\int_0^{s} \omega(u)\,du} \int_0^s \omega(s-u)\,\|\,\mathbf{z}^\top(u)\,\|_1\,du \right]\,ds.$$
Now, let us denote by $M(t)=\sup_{u\in[0,t]}\Vert \,\mathbf{z}^\top(u)\,\Vert_1$, which is trivially increasing as a function of $t$. Since $\omega(t)\geq0$ for all $t\geq0$ and $$\displaystyle\frac{1}{\displaystyle\int_0^{s} \omega(u)\, du}\int_0^s \omega(s-u)\,\, du=1,$$ 
we have
$$
\Vert\,\mathbf{z}^\top(t)\,\Vert_1 \leq(1-\lambda)\int_0^t e^{-(1-\lambda)(t-s)} \left[\frac{1}{\displaystyle\int_0^{s} \omega(u)\,du} \int_0^s \omega(s-u)\,\|\,\mathbf{z}^\top(u)\,\|_1 du \right]\,ds\leq (1-\lambda)\int_0^t e^{-(1-\lambda)(t-s)}\, M(s)\,ds.
$$
Now, since $[0,T]$ is a compact interval and $e^{-(1-\lambda)(t-s)}$ is a continuous function, let $C=\displaystyle\max_{0\leq s\leq t\leq T}e^{-(1-\lambda)(t-s)}$. Then
\begin{equation}\label{eq:inequality-constant-vectorial}
\Vert\,\mathbf{z}^\top(t)\,\Vert_1 \leq  (1-\lambda)\int_{0}^te^{-(1-\lambda)(t-s)}\,M(s)\,ds\leq (1-\lambda)\int_{0}^tC\,M(s)\,ds.
\end{equation}
Notice that the right-hand side is the integral of the non-negative function $M(s)$, and consequently is non-decreasing with respect to $t$. Moreover, since the integral of $M(s)$ acts as an upper bound for $\Vert \,\mathbf{z}^\top(t)\,\Vert_1$ at every instant $t$ and it never decreases, it also bounds the maximum value of $\Vert\,\mathbf{z}^\top(u)\,\Vert_1$ over the interval $[0,t]$. Therefore, we can replace the left-hand side $\Vert \,\mathbf{z}^\top(t)\,\Vert_1$ by its maximum obtaining the inequality
\begin{equation}\label{eq:inequality-gronwall-vectorial}
 M(t) \leq (1-\lambda)\int_{0}^tC\,M(s)\,ds,\quad 0\leq t\leq T.
\end{equation}
We just need to use Gronwall's inequality above (with  $\alpha=0$)  to conclude that $M(t)=0$, that is, $\mathbf{z}(t)=\mathbf{0}_{N\times1}$ for all $t\in [0,T]$.\\
\end{proof}

Next we will show that the the solution $\mathbf{x}(t)$ of the continuous-time linear system given by equation (\ref{eq:IVP}) is a non-negative probability vector. Notice that $\mathbf{x}(t)$ is non negative if and only if $\mathbf{x}^\top(t)\,\mathbf{e}_i\geq0$ for all $i=1,\dots, N$, for all time $t\geq0$,  where $\mathbf{e}_i$ is the $i^{th}$-unit vector of $\mathbb{R}^{N\times1}$. Additionally  we will show that  $ \mathbf{x}^\top(t)\,\mathbf{e}=1$, for all time $t\geq0$. Both facts will be obtained as two separate results. \\

In order to prove the non-negativity of the vector $\mathbf{x}(t)$, we need to recall what a Metzler matrix is.\\

\begin{definition}{\cite[Theorem 2, Section 2, Part I]{Farina-Rinaldi}}
A real square matrix $A=(a_{ij})\in \R^{N\times N}$ is said to be a Metzler matrix if its all off-diagonal entries are non-negative, i.e., $a_{ij}\geq0$ for all $i\neq j$.     
\end{definition}

With this definition at hand, we proceed to show that the solution $\mathbf{x}(t)$ of the IVP in equation (\ref{eq:IVP}) is a non-negative vector. 

\begin{thm}\label{thm:positivity}
Let $A$ be a non-negative square matrix of order $N$ associated with a directed graph $\mathcal{G}$, and let $P_A$ denote its row-normalization described in Section~\ref{Section:Notation}. Let $\mathbf{x}_0\in\mathbb{R}^{N\times1}$ be a probability vector (i.e., $\mathbf{x}_0\ge0$ and $\mathbf{x}_0^\top\mathbf{e}=1$), and let $\lambda\in(0,1)$ be a fixed damping factor. Suppose that $\omega:[0,\infty)\to[0,\infty)$ is a locally integrable function. If $\mathbf{x}(t)$, $t\ge0$, denotes the unique solution of the IVP
\begin{equation}\label{eq:probability-vector-1}
\dot{\mathbf{x}}^\top(t)=\mathbf{x}^\top(t)(\lambda P_A-I_N)+(1-\lambda)\left[\frac{1}{\displaystyle\int_0^t \omega(u)\,du}\int_0^t \omega(t-u)\,\mathbf{x}^\top (u)\,du\right],\quad \mathbf{x}(0)=\mathbf{x}_0,
\end{equation}
then $\mathbf{x}(s)\ge0$ for all $s\ge0$.
\end{thm}

\begin{proof}
Let us denote by $\mathbf{x}(t)$ the solution (at time $t$) for the IVP described in equation (\ref{eq:probability-vector-1}). To show the non-negativity of the vector $\mathbf{x}(t)$, we follow the ideas in \cite[Theorem 2, Section 2, Part I]{Farina-Rinaldi}; accordingly, in order to prove that $\mathbf{x}(s)\geq0$ for all $s\geq0$, it suffices to check that $\dot{\mathbf{x}}^\top(s)\geq0$ whenever $\mathbf{x}^\top(s)$ is on the boundary of the orthant $\mathcal{O}_+=\{\mathbf{x}^\top=(x_1,\dots,x_N)\,:\, x_i\geq0\text{ for }i=1,\dots, N\}$. Notice that equation (\ref{eq:probability-vector-1}) can be rewritten as follows 
\begin{equation}\label{eq:probability-vector-1-1}
\dot{\mathbf{x}}^\top(t)=\mathbf{x}^\top(t)(\lambda P_A-I_N)+(1-\lambda)\left[\frac{1}{\displaystyle\int_0^t \omega(u)\,du}\int_0^t \omega(t-u)\,\mathbf{x}^\top (u)\,du\right]=\mathbf{x}^\top(t)\,\mathbf{\Omega}+(1-\lambda)\mathbf{G}(t,\omega,\mathbf{x}^\top),
\end{equation}
where $\mathbf{\Omega}=(\lambda P_A-I_N)$ and $G(t,\omega,\mathbf{x}^\top)$ is the convolution part of the equation against $\omega$. 
In light of equation (\ref{eq:probability-vector-1-1}), we observe that showing $\dot{\mathbf{x}}^\top(s)\geq0$ is equivalent to showing that the components of the vector $\dot{\mathbf{x}}^\top(t)=\mathbf{x}^\top(t)\,\mathbf{\Omega}+(1-\lambda)\mathbf{G}(t,\omega,\mathbf{x})$ corresponding to the zero components of $\mathbf{x}^\top(t)$ are non-negative (i.e., if $x_i(t)=0$ for some index $i$, then $\dot{x}_i(t)\geq0$). Notice also that  $\mathbf{\Omega}=(\lambda P_A-I_N)$ is a Metzler matrix, since its off-diagonal entries, $\lambda p_{ij}$ with $P_A=(p_{ij})$, are non-negative.\\

Let us denote by $t^*$ the first instant for which some components of $\mathbf{x}^\top(t^*)$ are equal to zero and let $I^*$ be the set of indexes of such components, i.e., $I^*=\{i\in\{1,\dots,N\}\,:\, x_i(t^*)=0 \text{ with }\mathbf{x}^\top(t^*)=(x_1(t^*),\dots,x_N(t^*))\}$. Multiplying equation (\ref{eq:probability-vector-1-1}) on the right by the column vector $\mathbf{e}_j$ (where its $j^{th}$ component is equal to one and zero for the rest) 
$$\dot{\mathbf{x}}^\top(t^*)\,\mathbf{e}_j=\mathbf{x}^\top(t^*)\mathbf{\Omega}\,\mathbf{e}_j+(1-\lambda)\mathbf{G}(t^*,\omega,\mathbf{x}^\top)\,\mathbf{e}_j$$
we obtain the following integro-differential scalar equation
\begin{align*}
\dot{x}_j(t^*)
&=\sum_{i\notin I^*}a_{ij}x_i(t^*)+(1-\lambda)\left[\frac{1}{\displaystyle\int_0^{t^*} \omega(u)\,du}\int_0^{t^*} \omega(t^*-u) x_j(u)\,du\right].
\end{align*}
Therefore, in view of  $a_{ij}\geq0$, $x_i(t^*)>0$ for $i\notin I^*$, and also that the weight function $\omega$ is non-negative and $x_j(s)>0$ for all $s<t^*$ (since $t^*$ is the first instant such that $x_j(t^*)=0$), it follows that that $\dot{x}_j(t^*)\geq0$.
\end{proof}

\begin{remark}
A similar proof for the positivity of $\mathbf{x}(s)$ applies in the case $\mathbf{x}_0>0$, by considering the existence of a time $t^*>0$ as the first instant at which some components of $\mathbf{x}(t^*)$ are zero. More precisely, if we denote by $I^*$ the set of indices such that of $x_j(t^*)=0$, then, proceeding similarly, it can be proved that $\dot{x}_j(t^*)\geq0$.
\end{remark}

When it comes to proving  that the normalization condition $\mathbf{x}^\top(s)\, \mathbf{e}=1$ holds for all $s\geq0$, we need to recall the notion of uniform equicontinuity.

\begin{definition}
Let $K$ be a compact metric space and let denote by $\mathcal{C}(K)$ the space of continuous function on $K$. A bounded subset $\mathcal{H}$ of $\mathcal{C}(K)$ is said to be uniformly equicontinuous if for every $\varepsilon>0$, there exists  $\delta>0$ such that if $\vert t_1-t_2\vert\leq \delta$, then $\vert f(t_1)-f(t_2)\vert<\varepsilon$ for all $f\in\mathcal{H}$ and all $t_1,t_2\in K$.
\end{definition}

With this definition at hand, we are in position to show the normalization condition on the solution of the continuous dynamical system with an initial condition.

\begin{thm}\label{thm:identity-to-one}
Let $A$ be a non-negative square matrix of order $N$ associated with a directed graph $\mathcal{G}$, and let $P_A$ denote its row-normalization described in Section~\ref{Section:Notation}. Let $\mathbf{x}_0\in\mathbb{R}^{N\times1}$ be a probability vector (i.e., $\mathbf{x}_0\ge0$ and $\mathbf{x}_0^\top\mathbf{e}=1$), and let $\lambda\in(0,1)$ be a fixed damping factor. Suppose that $\omega:[0,\infty)\to[0,\infty)$ is a locally integrable function. If $\mathbf{x}(t)$, for $t\ge0$, denotes the unique solution of the IVP
\begin{equation}\label{eq:probability-vector-2}
\dot{\mathbf{x}}^\top(t)=\mathbf{x}^\top(t)(\lambda P_A-I_N)+(1-\lambda)\left[\frac{1}{\displaystyle\int_0^t \omega(u)\,du}\int_0^t \omega(t-u)\,\mathbf{x}^\top (u)\,du\right],\quad \mathbf{x}(0)=\mathbf{x}_0,
\end{equation}
then $\mathbf{x}^\top(s)\,\mathbf{e}=1$ for all $s\geq0$.
\end{thm}

\begin{proof}
Let us denote by $\mathbf{x}(t)$ the solution (at time $t$) for the IVP described in equation (\ref{eq:probability-vector-2}). To show that $\mathbf{x}^\top(s)\,\mathbf{e}=1$ for all $s\geq0$, let us define the scalar function $y(s)=\mathbf{x}^\top(s)\,\mathbf{e}$ for $s\geq0$. If we multiply equation (\ref{eq:probability-vector-2}) on the right by the vector $\mathbf{e}=(1,\dots,1)^\top$, using the row-stochasticity of the matrix $P_A$ we obtain a scalar integro-differential equation given by
\begin{equation}\label{eq:scalar-diff}
 \dot{y}(t)=-(1-\lambda)y(t)+(1-\lambda)\left[\frac{1}{\displaystyle\int_0^t \omega(u)\,du}\int_0^t \omega(t-u)\,y(u)\,du\right]=f(t,\,y(t)).
\end{equation}
Now, assuming that $\mathbf{x}(0)$ is a probability vector, we have the initial condition $y(0)=\mathbf{x}^\top(0)\,\mathbf{e}=1$. With equation (\ref{eq:scalar-diff}) and the initial condition $y(0)=1$ at hand, the idea is to show the existence and uniqueness of a (local) solution for the IVP $\dot{y}(t)=f(t,\,y(t))$ with $y(0)=1$ using a classical Fixed-Point Theorem approach. Once we have this local solution, it will be extended to all $t\geq0$ using a continuation of solutions argument (see \cite[Chapter 1, Section 4, Theorem 4.1]{Coddington-Levinson} or \cite[Chapter 8, Section 8.5, Corollary 8.35]{Kelley-Peterson}, for instance).\\

It is clear the function $y(t)=1$ for $t\geq0$ is a solution of the equation (\ref{eq:scalar-diff}) which satisfies the initial condition $y(0)=1$. Now, drawing on classical fixed-point theorems, we show that the IVP admits no other solution. For the sake of completeness, we include a proof based on Schauder’s fixed-point theorem (see \cite[Chapter 12, Theorem 1.4]{GLS}, for instance), following the ideas presented in \cite[Chapter 12]{GLS}. A straightforward observation is that a function $\varphi$ is a solution of the IVP  $y'(t)=f(t,\,y(t))$ with $y(0)=1$ if and only if $\varphi$ satisfies the integral equation
$$\varphi(t)=1+\int_0^t f(s,\, \varphi(s))\,ds,\quad t\geq0.$$
First, we prove that the previous equation has a solution on the interval $[0,T]$ where $T>0$ is sufficiently small. To this aim, since the function $f(t,y(t))$ in equation (\ref{eq:scalar-diff}) is undefined at $t=0$, we consider  the operator $T_f$ defined on the space of continuous functions $\mathcal{C}([0,T])$ given by 
\begin{equation}\label{eq:definition-operator}
 T_f\phi(t)=
\begin{cases}
1, & \text{for } t=0\\
1+\displaystyle\int_0^t f(s,\, \phi(s))\,ds, &\text{for } 0<t\leq T,
\end{cases}
\end{equation}
and show that it has a fixed-point, that is, a function $\varphi$ such that $T_f\varphi(t)=\varphi(t)$ for all $t\in[0,T]$. To this end we consider as the  domain of $T_f$ the set 
$$K=\{\phi\in\mathcal{C}([0,T])\,:\,\phi(0)=1 \text{ and }\vert\phi(t)\vert\leq M\text{ for all } t \in [0,T]\},$$
where $M$ is some positive real number. In fact, the number $T>0$ and $M>0$ are chosen so as to ensure that the operator $T_f$ maps the set $K$ into itself, which is a standard requirement in fixed-points theorems. 

It is clear that the function $t\to T_f\phi(t)$ is continuous for all $0<t\leq T$. {Moreover, it can be proved that this mapping is also continuous at $t=0$, that is, $T_f\phi(0)=\displaystyle\lim_{t\to0^+}T_f\phi(t)$}. Firstly, from its definition in (\ref{eq:definition-operator}) we have $T_f\phi(0)=1$. On the other hand, since the weight function $\omega$ is positive and $\phi\in K$, for all $t\in(0,T]$ we have
\begin{align}\label{eq:bounded-f}
\vert f(t,\phi(t))\vert
&=\left\vert -(1-\lambda)\phi(t) + (1-\lambda)\frac{1}{\displaystyle\int_0^t \omega(u)\,du}\int_0^t \omega(t-u)\,\phi(u)\,du\right\vert\nonumber\\ 
&\leq (1-\lambda)\vert\phi(t)\vert  + (1-\lambda)\frac{1}{\displaystyle\int_0^t \omega(u)\,du}\int_0^t \omega(t-u)\,\vert \phi(u)\vert \,du\nonumber\\
&\leq (1-\lambda)M  + (1-\lambda)M\left[\frac{1}{\displaystyle\int_0^t \omega(u)\,du}\int_0^t \omega(t-u)\,\,du\right]=2(1-\lambda)M,
\end{align}
which implies 
$$\lim_{t\to0^+}\left\vert\enspace\int_0^tf(s,\phi(s))\,ds\enspace\right\vert\leq\lim_{t\to0^+}\int_0^t \vert f(s,\, \phi(s))\vert\,ds\leq \lim_{t\to0^+}2(1-\lambda)Mt=0.$$
Therefore, $T_f\phi(0)=1=\displaystyle\lim_{t\to0^+}T_f\phi(t)$. {In fact, it can also be proved that the function $T_f\phi(t)$ is Lipschitz continuous on $[0,T]$.}

Since $T_f\phi(0)=1$, in order to prove that $T_f$ maps the set $K$ into itself, it suffices to show that $\vert T_f\phi(t)\vert \leq M$ for all $t\in[0,T]$. Using the boundedness in equation (\ref{eq:bounded-f}), for all $t\in [0,T]$ we have
$$\vert T_f\phi(t)\vert=\left \vert 1+\int_0^t f(s,\, \phi(s))\,ds\right\vert\leq 1+\int_0^t \vert f(s,\, \phi(s))\vert\,ds\leq 1+2(1-\lambda)MT.$$
At this point, if we take $M=2$ and $T\leq 1$ sufficiently small, then $\vert T_f\phi(t)\vert\leq M$ for all $t\in[0,T]$. Hence, $T_f$ maps $K$ into itself.\\

Now, it is clear that $K$ is a closed, bounded and convex subset of the Banach space $\mathcal{C}([0,T])$. In order to apply the Schauder's fixed-point theorem, we need to show that the operator $T_f$ is a compact mapping. To this aim, by the Arzelà-Ascoli Theorem (see \cite[Chapter 4, Section 4.5, Theorem 4.25]{Brezis}), it suffices to show that $T_f(K)=\{T_f(\phi)\,:\,\phi\in K\}$ is a uniformly equicontinuous set of $\mathcal{C}([0,T])$. To this aim, if we take $t_1,t_2\in [0,T]$ with $0\leq t_1\leq t_2\leq T$, for all $\phi\in K$ the boundedness in equation (\ref{eq:bounded-f}) implies 
\begin{equation}\label{eq:equconinuous}
\vert\enspace(T_f\phi)(t_1)-(T_f\phi)(t_2)\enspace\vert
\leq\int_{t_1}^{t_2} \vert f(s,\, \phi(s))\vert\,ds\leq 2(1-\lambda)M(t_2-t_1).
\end{equation}
Therefore, with equation (\ref{eq:equconinuous}) at hand, we conclude that the set $T_f(K)$ is a uniformly equicontinuous set of $\mathcal{C}([0,T])$ and therefore $T_f(K)=\{T_f(\phi)\,:\,\phi\in K\}$ is a compact set. As an application of the Schauder's fixed-point theorem, there exists a function $\varphi\in\mathcal{C}([0,T])$ (not necessarily unique) such that $T_f\varphi(t)=\varphi(t)$ for all $t\in[0,T]$, that is, a solution for the IVP $\dot{y}(t)=f(t,\,y(t))$ with $y(0)=1$. 

Now, the next step is to prove that this solution is unique. To this aim, we use similar arguments to the one used in the proof of Theorem \ref{thm:existence-uniqueness}. A straightforward observation is that the constant function $y(t)=1$ with $t\in[0,T]$ is a solution for the IVP in equation (\ref{eq:scalar-diff}). Let $y$ be a function defined on $[0,T]$ which is another solution of the IVP in equation (\ref{eq:scalar-diff}) and let us define the difference function $z(t)=y(t)-1$, for $t\in[0,T]$. We need to prove that $z(t)=0$ for all $t\in[0,T]$. By linearity, it is clear that the function $z(t)$ satisfies the integro-differential equation
\begin{equation}\label{eq:scalar-unique}
 \dot{z}(t)+(1-\lambda)x(t)=(1-\lambda)\left[\frac{1}{\displaystyle\int_0^t \omega(u)\,du}\int_0^t \omega(t-u)\,z(u)\,du\right],
\end{equation}
with the initial condition $z(0)=y(0)-1=0$ (since $y(0)=1$). Now, if we multiply equation (\ref{eq:scalar-unique}) by the integrating factor $e^{(1-\lambda)t}$ we have
\begin{align*}
e^{(1-\lambda)t}\,\dot{z}(t)+(1-\lambda)e^{(1-\lambda)t}z(t)&=(1-\lambda)e^{(1-\lambda)t}\left[\frac{1}{\displaystyle\int_0^t \omega(u)\,du}\int_0^t \omega(t-u)\,z(u)\,du\right],\\
\frac{d}{dt}\left[e^{(1-\lambda)t}z(t)\right]&=(1-\lambda)e^{(1-\lambda)t}\left[\frac{1}{\displaystyle\int_0^t \omega(u)\,du}\int_0^t \omega(t-u)\,z(u)\,du\right],\\
\end{align*}
Now, since $z(0)=0$, integrating both sides from $0$ to $t$ in the previous equation and multiplying by the factor $e^{-(1-\lambda)t}$ we get the equality
\begin{equation}\label{eq:isolate-x}
z(t)=(1-\lambda)\int_{0}^te^{-(1-\lambda)(t-s)}\left[\frac{1}{\displaystyle\int_0^s \omega(u)\,du}\int_0^s \omega(t-u)\,z(u)\,du\right]\, ds.  
\end{equation}
Now, let denote by $M(t)=\sup_{u\in[0,t]}\vert\,z(u)\,\vert$. Since $\omega(t)\geq0$ for all $t\geq0$ and $\frac{1}{\int_0^s \omega(u)\, du}\int_0^s \omega(s-u)\,\, du=1$, we have
\begin{equation}\label{eq:uniqueness-bound}
\left\vert \frac{1}{\displaystyle\int_0^s \omega(u)\,du}\int_0^s \omega(s-u)\,x(u)\,du\right\vert 
\leq  \frac{1}{\displaystyle\int_0^s \omega(u)\,du}\int_0^s \omega(s-u)\,\vert\, z(u)\,\vert\,du
\leq M(s).
\end{equation}
At this point, since the damping factor $\lambda\in(0,1)$, we can take absolute values in equation (\ref{eq:isolate-x}) and use the boundedness in equation (\ref{eq:uniqueness-bound}) to get the inequality
\begin{equation}\label{eq:inequality-uniqueness}
\vert\, z(t)\,\vert  \leq  (1-\lambda)\int_{0}^te^{-(1-\lambda)(t-s)}M(s)\,ds.
\end{equation}
Now, since $[0,T]$ is a compact interval and $e^{-(1-\lambda)(t-s)}$ is a continuous function, let $C=\max_{0\leq s\leq t\leq T}e^{-(1-\lambda)(t-s)}$. With this constant at hand, we get
\begin{equation}\label{eq:inequality-constant}
\vert\, z(t) \,\vert  \leq  (1-\lambda)\int_{0}^te^{-(1-\lambda)(t-s)}M(s)\,ds\leq (1-\lambda)\int_{0}^tC\,M(s)\,ds.
\end{equation}
Notice that the right-hand side is the integral of the non-negative function $M(s)$, which implies that it is non-decreasing with respect to $t$. Moreover, since the integral of $M(s)$ acts as an upper bound for $\vert\, z(t)\,\vert$ at every instant $t$ and it never decreases, it also bounds the maximum value of $\vert\, z(u)\,\vert$ over the interval $[0,t]$. Therefore, we can replace the left-hand side $\vert\, z(u)\,\vert$ by its maximum obtaining the inequality
\begin{equation}\label{eq:inequality-gronwall}
 M(t) \leq (1-\lambda)\int_{0}^tC\,M(s)\,ds,\quad 0\leq t\leq T.
\end{equation}
We just need to use Gronwall's inequality above (with $\alpha=0$) to conclude that $M(t)=0$, that is, $z(t)=0$, for all $t\in [0,T]$.\\

Finally, the solution can be extended beyond the interval $[0,T]$ as follows: Since $\varphi(t)=1$ is the solution of the IVP on $[0,T]$, let us propose the IVP 
\begin{equation}\label{eq:extension-solution}
\dot{y}(t)=f(t,\,y(t)),\quad t\geq T,\quad \text{with the initial condition}\quad  y(T)=\varphi(T)=1.   
\end{equation} 
Now, by considering the Banach space $\mathcal{C}([T,S])$ for some $0<T<S$ small enough, using similar arguments as above, there exists a unique function $\varphi_1\in\mathcal{C}([T,S])$ to be the solution to the IVP in equation (\ref{eq:extension-solution}). At this point, since the constant function $y(t)=1$ with $t\geq0$ (in particular, for $t\in[T,S]$) is a solution for the same IVP in equation (\ref{eq:extension-solution}), we conclude that $\varphi_1(t)=1$ for all $t\in[T,S]$. Therefore, the solution $\varphi(t)=1$ can be extended from $[0,T]$ to $[0,S]$. Since the solution for the IVP on the corresponding interval remains bounded by the arguments above, the function $ \varphi(t)=1$ admits an extension to the maximal interval $[0,\infty)$, and the proof that $y(s)=1$ for $s\geq0$ is the unique solution of the scalar integro-differential equation (\ref{eq:scalar-diff}) with $y(0)=1$ is completed.  

\end{proof}

In addition, for the solution to the IVP described in equation (\ref{eq:probability-vector-1}), we establish a conservation law, namely that $\mathbf{x}(t)$ is a probability vector for all $t\geq0$ if and only if the sum of the components of its derivative $\dot{\mathbf{x}}(t)$ is zero for all $t\geq0$. More precisely,

\begin{proposition}\label{prop:conservation-law}
Let $A$ be a non-negative square matrix of order $N$ associated with a directed graph $\mathcal{G}$, and let $P_A$ denote its row-normalization described in Section~\ref{Section:Notation}. Let $\mathbf{x}_0\in\mathbb{R}^{N\times1}$ be a probability vector (i.e., $\mathbf{x}_0\ge0$ and $\mathbf{x}_0^\top\mathbf{e}=1$), and let $\lambda\in(0,1)$ be a fixed damping factor. Suppose that $\omega:[0,\infty)\to[0,\infty)$ is a locally integrable function. If $\mathbf{x}(t)$ denotes the unique solution of the IVP
\begin{equation}\label{eq:conservation-law}
\dot{\mathbf{x}}^\top(t)=\mathbf{x}^\top(t)(\lambda P_A-I_N)+(1-\lambda)\left[\frac{1}{\displaystyle\int_0^t \omega(u)\,du}\int_0^t \omega(t-u)\,\mathbf{x}^\top (u)\,du\right],\quad \mathbf{x}(0)=\mathbf{x}_0,
\end{equation}
then, $\mathbf{x}(t)^\top\,\mathbf{e}=1$ for all $t\geq0$ if and only if $\dot{\mathbf{x}}(t)^\top\,\mathbf{e}=0$ for all $t\geq0$.
\end{proposition}

\begin{proof}
Let $t_0\geq0$ be a fixed instant. From the definition of the derivative we have
$$\dot{\mathbf{x}}^\top(t_0)=\lim_{h\to 0}\frac{\mathbf{x}^\top(t_0+h)-\mathbf{x}^\top(t_0)}{h}.$$
Assuming that $\mathbf{x}^\top(t)\,\mathbf{e}=1$ for all $t\geq0$, by multiplying the above equation on the right by the vector $\mathbf{e}=(1,\dots,1)^\top$ we get
$$\dot{\mathbf{x}}^\top(t_0)\,\mathbf{e}=\left(\lim_{h\to 0}\frac{\mathbf{x}^\top(t_0+h)-\mathbf{x}^\top(t_0)}{h}\right)\,\mathbf{e}=\lim_{h\to 0}\frac{\mathbf{x}^\top(t_0+h)\,\mathbf{e}-\mathbf{x}^\top(t_0)\,\mathbf{e}}{h}=0.$$
Since this argument is valid for every $t_0\geq0$, we conclude that $\dot{\mathbf{x}}^\top(t)\,\mathbf{e}=0$ for all $t\geq0$.\\

Conversely, suppose that $\dot{\mathbf{x}}^\top(t)\,\mathbf{e}=0$ for all $t\geq0$ and let us to consider a fixed instant $t_0\geq0$. Since $\mathbf{x}$ is a continuous and differentiable function (by its proper definition by equation (\ref{eq:conservation-law}), by the Mean Value Theorem there exists a $\xi\in(0,t_0)$ such that
$$\dot{\mathbf{x}}^\top(\xi)=\frac{\mathbf{x}^\top(t_0)-\mathbf{x}^\top(0)}{t_0}.$$
By multiplying the above equation on the right by the vector $\mathbf{e}=(1,\dots,1)^\top$ and using the hypothesis, we have
$$0=\dot{\mathbf{x}}^\top(\xi)\,\mathbf{e}=\left(\frac{\mathbf{x}^\top(t_0)-\mathbf{x}^\top(0)}{t_0}\right)\mathbf{e}=\frac{\mathbf{x}^\top(t_0)\,\mathbf{e}-\mathbf{x}^\top(0)\,\mathbf{e}}{s_0}=\frac{\mathbf{x}^\top(t_0)\,\mathbf{e}-1}{t_0}.$$
Therefore, $\mathbf{x}^\top(t_0)\,\mathbf{e}=1$. Since this argument is valid for every $t_0>0$, we conclude that $\mathbf{x}^\top(t)\,\mathbf{e}=1$ for all $t\geq0$.\\
\end{proof}

In the next section, we analyze the asymptotic behavior of the solution $\mathbf{x}(t)$ of the dynamical system for the PageRank with time-dependent memory  on a strongly connected digraph $\mathcal{G}$, considering several classes of locally integrable weight functions $\omega:[0,\infty)\to[0,\infty)$, damping factors $\lambda\in(0,1)$, and initial conditions $\mathbf{x}_0\in \mathbb{R}^{N \times 1}$.

\section{The asymptotic behavior of the solution}
\label{Section:Asymptotic}

In this section we investigate the asymptotic behavior of the solution for the integro-differential equation with the initial problem 
\begin{equation}\label{eq:IVP-asymptotic}
\dot{\mathbf{x}}^\top(t)
=\mathbf{x}^\top(t)(\lambda P_A-I_N)+(1-\lambda)\left[\frac{1}{\displaystyle\int_0^t \omega(u)\,du}\int_0^t \omega(t-u)\,\mathbf{x}^\top (u)\,du\right],\quad \mathbf{x}(0)=\mathbf{x}_0,    
\end{equation}
where $P_A$ is the row-normalization of the adjacency matrix $A$ associated with the directed graph $\mathcal{G}$, $\lambda\in(0,1)$ is the damping factor, $\omega:[0,\infty)\to[0,\infty)$ is a locally integrable weight function (sometimes called memory function), and the probability vector $\mathbf{x}_0\in\mathbb{R}^{N\times1}$ as a initial condition. \\

For the case where $A$ is an irreducible square matrix of order $N$, we show that the solution $\mathbf{x}(t)$ of the IVP proposed in equation (\ref{eq:IVP-asymptotic}) converges asymptotically (i.e., as $t\to\infty$) to the left-hand Perron vector (see the definition below) of the row-normalized matrix $P_A$, independently of the damping factor $\lambda\in(0,1)$ and the initial condition $\mathbf{x}_0\in\mathbb{R}^{N\times1}$ (see Theorem \ref{thm:delta-case} and Theorem \ref{thm:exponential-case}). Additionally, for a particular oscillatory memory function $\omega$, we prove that the solution $\mathbf{x}(t)$ of the IVP proposed in equation (\ref{eq:IVP-asymptotic}) is asymptotically periodic (see Theorem \ref{thm:oscillatory-case}). Moreover, in this case, this asymptotically periodic solution exhibits a strong dependence on the initial condition $\mathbf{x}_0\in\mathbb{R}^{N\times1}$ and differs significantly from the left-hand Perron vector of the matrix $P_A$, in contrast to the previous cases. Before stating these theorems, let us recall some useful definitions and results.\\

A square matrix $A$ of order $N\geq2$ is said to be \em reducible \em if there exists a permutation matrix $R$ such
that 
\begin{equation*}
R^TAR=\left(\begin{array}{c|c} 
    X & Y \\ 
    \hline
    \bold{0} & Z \end{array}\right),
\end{equation*}
where $X$ and $Z$ are both square matrices (see \cite[Section 4.4]{Meyer}). Otherwise, $A$ is said to be \em irreducible\em . A directed graph $\mathcal{G}=(V,E)$ of $N$ nodes is \em strongly connected \em if for any ordered pair $(i,j)\in V\times V$, $1\leq i,j\leq N$, there exists a directed path of edges in $E$ leading from node $i$ to node $j$ (see \cite[Section 1]{Varga}). In  fact, the following result establishes an equivalence between the irreducibility  of a matrix in terms of its strong connectivity. More precisely,

\begin{lem}(\cite[Theorem 1.17]{Varga}){\label{lemma:strongly-connected-irreducible}}
An $N\times N$  matrix $A$ is irreducible if and only if its directed graph $\mathcal{G}$ associated to $A$ is strongly connected.\\
\end{lem}

Notice that a non-negative matrix $A\neq0$ is irreducible if and only if for every $1\leq i,j\leq N$, there exists a positive integer $k\geq 1$ such that the $(i,j)$-entry of $A^k$ is positive (see \cite[Chapter 2, Theorem 2.1]{BemannPlemmons}).\\

In the context of the Perron-Frobenius theory, if $A\in\R^{N\times N}$ is a non-negative and irreducible matrix, there exists a unique vector $\bold{q}\in\R^{N\times1}$, called the \textit{left-hand Perron vector of the matrix $A$}, such that 
\begin{equation*}
    \bold{q}^TA=r\bold{q}^T\,\qquad\text{with}\qquad \, \bold{q}>0 \qquad\text{and}\qquad \bold{q}^T\bold{e}=1,
\end{equation*}
where $r=\rho(A)$ is the spectral radius of the matrix $A$ (see \cite[Section 8.3]{Meyer}). \\

In the following subsections, we illustrate through numerical simulations the asymptotic behavior of the components of the solution $\mathbf{x}(t)$ of the IVP in equation \eqref{eq:IVP-asymptotic}, as well as its dependence on the damping factor and the initial condition. These simulations will be made for several classes of memory functions such as the Dirac delta function, the exponential decay function, and exponentially oscillatory functions. For this purpose,  two strongly connected directed graphs with three nodes will be considered: the $3$-cycle graph $\mathcal{G}_1$ in Figure \ref{fig:numerical-examples}(a) and the graph $\mathcal{G}_2$ in Figure \ref{fig:numerical-examples}(b), which has bidirectional edges between nodes $1$-$2$ and $2$–$3$. 

 \begin{figure}[H]
        \centering
        \includegraphics[width=0.6\textwidth]{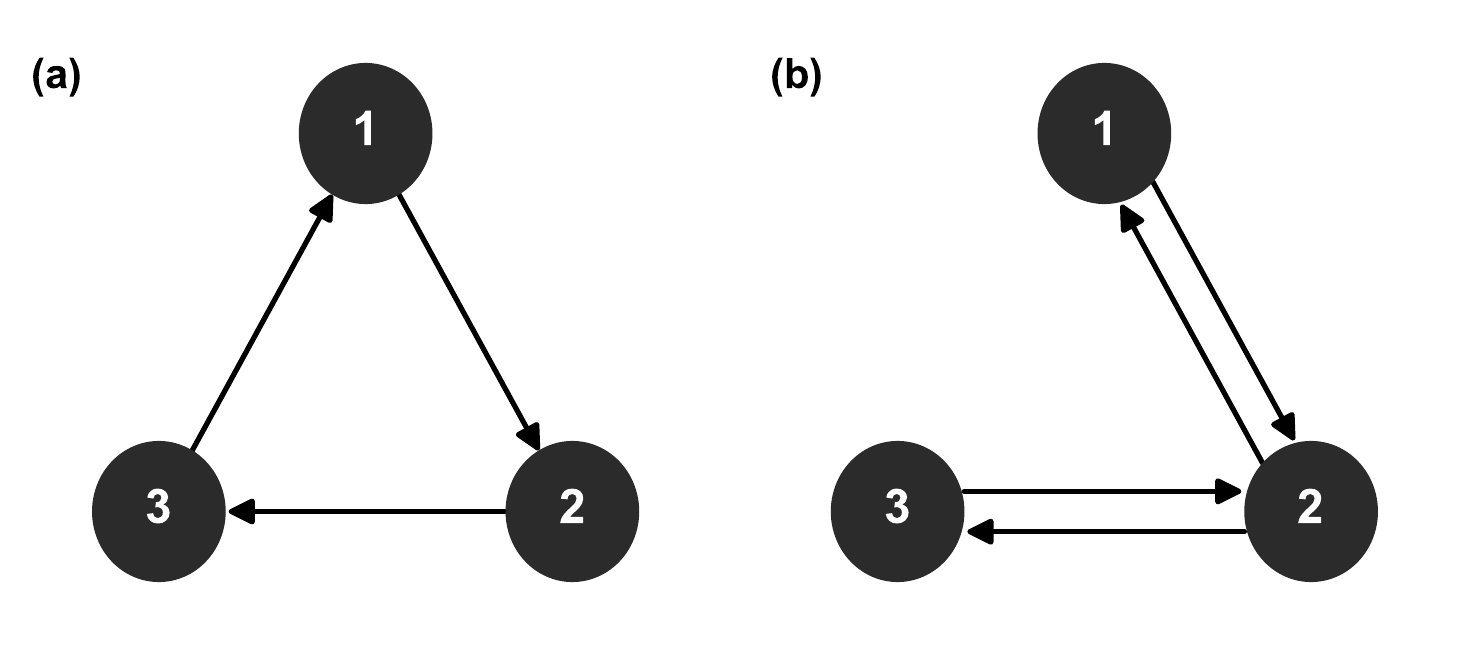}
        \caption{The two strongly connected digraph used for numerical simulations in the following subsections.}
        \label{fig:numerical-examples}
    \end{figure}

It is worth mentioning that, for both the $3$-cycle in Figure \ref{fig:numerical-examples}(a) and the digraph which has bidirectional edges between nodes 1-2 and 2–3 in Figure \ref{fig:numerical-examples}(b), we tested several values of the damping factor and different initial conditions, and the same asymptotic behavior was observed in all cases.


\subsection[]{The Dirac delta memory function $\omega(t)=\delta_0(t)$, for $t\geq0$.}\label{subsect:asymptotic-delta}

In this subsection we will be considering as the weight function $\omega:[0,\infty)\to[0,\infty)$  the Dirac delta function, that is, $\omega(t)=\delta_0(t)$ for $t\geq0$. This choice corresponds to the case in which the system has no distributed memory effects. Indeed, the Dirac delta assigns all the weight to the present instant $t$, so that the dynamics governed at time $t$ depend exclusively on the current state of the system and not on any weighted contribution from its past states. 

Consequently, equation (\ref{eq:IVP-asymptotic}) reduces to a purely instantaneous interaction model, where no temporal averaging of the previous values of the vector $\mathbf{x}$ is involved. In this sense, the weight function $\omega(t)=\delta_0(t)$, for $t\geq0$, may be interpreted as an instantaneous memory kernel, since the system only retains information from the current time instant.\\

Notice that the Dirac delta function $\delta_0(t)$ can be considered as the pointwise limit of the characteristic functions $\chi_{\raisebox{-.5ex}{$\scriptstyle [0,\varepsilon]$}}(t)$ when $\varepsilon\to0^+$. Therefore, in this case the IVP described in equation (\ref{eq:IVP-asymptotic}) turns into
\begin{align}\label{eq:delta-case-numerical}
\dot{\mathbf{x}}^\top(t)=\mathbf{x}^\top (t)(\lambda P_A-I_N)+(1-\lambda)\left[\frac{1}{\varepsilon}\int_{t-\varepsilon}^t \mathbf{x}^\top(u)\,du\right],\quad \mathbf{x}(0)=\mathbf{x}_0,
\end{align}
where $P_A$ is the row-normalization of the adjacency matrix $A$ associated with the directed graph $\mathcal{G}$, $\lambda\in(0,1)$ is the damping factor, and the probability vector $\mathbf{x}_0\in\mathbb{R}^{N\times1}$ is the initial condition.\\

In order to investigate the asymptotic behavior of the solution for the IVP in equation (\ref{eq:delta-case-numerical}) above 
we perform numerical simulations for sufficiently small values of $\varepsilon>0$. As a first step, we consider the strongly connected directed 3-cycle graph shown in Figure \ref{fig:numerical-delta1}(a). To illustrate the behavior of the components of the solution $\mathbf{x}(t)=(\mathbf{x}_1(t),\mathbf{x}_2(t),\mathbf{x}_3(t))$ of the IVP in equation \eqref{eq:exponential-case}, we compute the corresponding numerical solution. In Figure \ref{fig:numerical-delta1}(b) we consider the damping factor $\lambda=0.5$ and the initial condition $\mathbf{x}_0^\top=(0.2,0.1,0.7)$, while in Figure \ref{fig:numerical-delta1}(c) we take  $\lambda=0.85$ and $\mathbf{x}_0^\top=(0.1,0.4,0.5)$.
    \begin{figure}[H]
        \centering
        \includegraphics[width=1\textwidth]{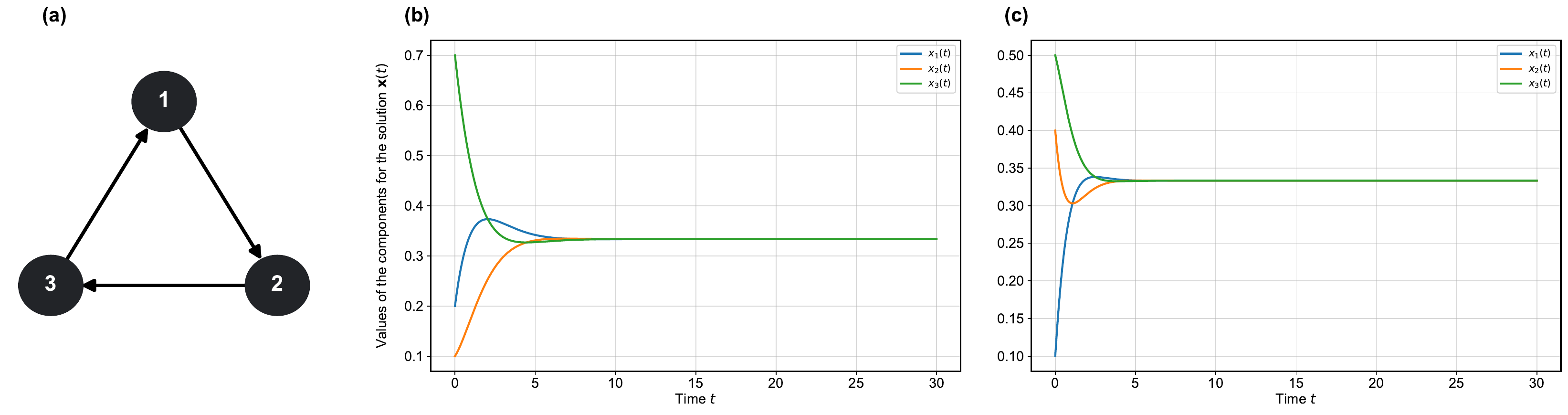}
        \caption{A numerical solution for the IVP in equation (\ref{eq:delta-case-numerical}) for the 3-cycle digraph and $\omega(t)=\delta_0 (t)$, for $t\geq0$.}
         \label{fig:numerical-delta1}
    \end{figure}
A straightforward observation from the numerical simulations shown in Figure \ref{fig:numerical-delta1} is that the solution $\mathbf{x}(t)$ is asymptotically stable. Moreover, it seems that the solution $\mathbf{x}(t)$ converges to the constant vector $(1/3,1/3,1/3)^\top\in\R^{3\times1}$ which is, in fact, the left-hand Perron vector of the row-normalization matrix $P_A$ associated with the adjacency matrix $A$ of the 3-cycle digraph in Figure \ref{fig:numerical-delta1}(a). \\

As a second step, we examine whether this phenomenon is specific to the $3$-cycle graph or also occurs in other strongly connected digraph, in Figure \ref{fig:numerical-delta2}(a) we consider a digraph which has bidirectional edges between nodes $1$-$2$ and $2$–$3$. To illustrate the behavior of the components of the solution $\mathbf{x}(t)=(\mathbf{x}_1(t),\mathbf{x}_2(t),\mathbf{x}_3(t))$ of the IVP in equation \eqref{eq:exponential-case}, we compute the corresponding numerical solution. As above, in Figure \ref{fig:numerical-delta2}(b) we consider the damping factor $\lambda=0.5$ and the initial condition $\mathbf{x}_0^\top=(0.2,0.1,0.7)$, while in Figure \ref{fig:numerical-delta2}(c) we consider $\lambda=0.85$ and $\mathbf{x}_0^\top=(0.8,0.1,0.1)$.
\begin{figure}[H]
        \centering
        \includegraphics[width=1\linewidth]{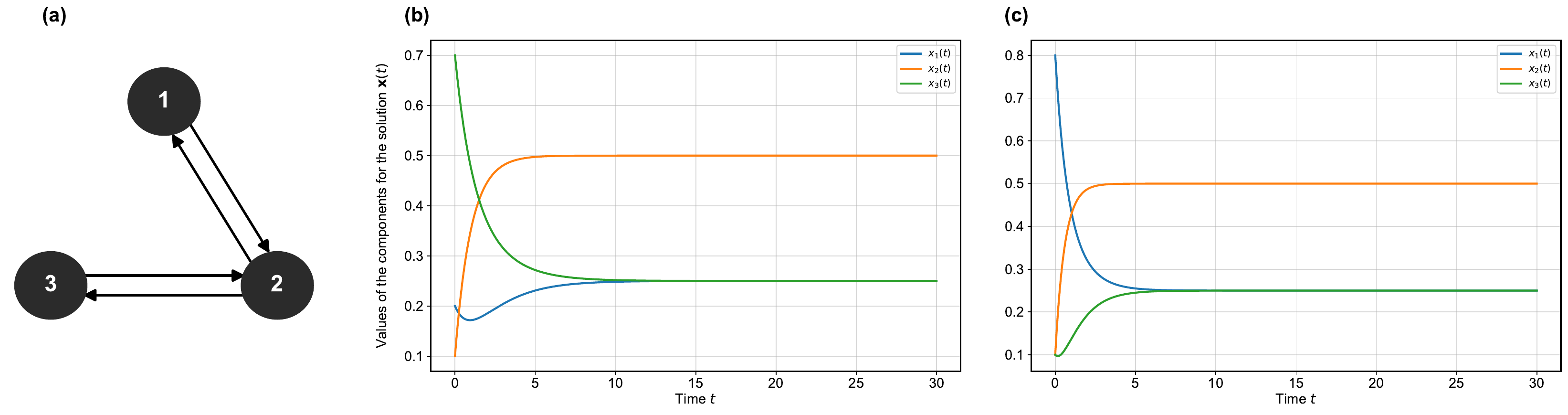}
        \caption{A numerical solution for the IVP in equation (\ref{eq:delta-case-numerical}) for the bidirectional digraph and $\omega(t)=\delta_0 (t)$, for $t\geq0$.}
        \label{fig:numerical-delta2}
    \end{figure}
Similarly to the $3$-cycle digraph, the numerical simulations shown in Figure \ref{fig:numerical-delta2} suggest that the solution $\mathbf{x}(t)$ is asymptotically stable and converges to the vector $(0.25,0.50,0.25)^\top\in\R^{3\times1}$. As before, this vector coincides with the left-hand Perron vector of the row-normalization matrix $P_A$ associated with the adjacency matrix $A$ for the digraph in Figure \ref{fig:numerical-delta2}(a). \\

In light of this numerical simulation 
we  conjecture that the solution of the IVP asymptotically stabilizes at the left-hand Perron vector of the matrix $P_A$, when the Dirac delta memory function is considered; in fact the conjecture is proved in Theorem \ref{thm:delta-case} below, where an explicit solution of the IVP in equation (\ref{eq:IVP-asymptotic}) is obtained; also we analyze its asymptotic behavior as $t\to\infty$. An important feature of this result is that the asymptotic limit is independent of both the damping factor $\lambda\in(0,1)$ and the choice of a probability vector $\mathbf{x}_0\in\mathbb{R}^{N\times1}$ as the initial condition. More precisely,\\


\begin{thm}\label{thm:delta-case}
Let $A$ be a non-negative irreducible square matrix of order $N$ of a directed graph $\mathcal{G}$ and let $P_A$ be its row-normalization described in Section \ref{Section:Notation}. Let $\omega(t)=\delta_0 (t)$, for $t\geq0$, be the non-negative and locally integrable function weight function. Then, for any personalization vector 
$\mathbf{x}_0\in\R^{N\times1}$ ($\mathbf{x}_0\geq0$ and $\mathbf{x}_0^\top\mathbf{e}=1$) and for any damping factor $\lambda\in(0,1)$, the solution $\mathbf{x}$ of the Initial Value Problem (IVP) given by
\begin{equation}\label{eq:IVP-delta-original}
\dot{\mathbf{x}}^\top(t)=\mathbf{x}^\top(t)(\lambda P_A-I_N)+(1-\lambda)\left[\frac{1}{\displaystyle\int_0^t \omega(u)\,du}\int_0^t \omega(t-u)\,\mathbf{x}^\top(u)\,du\right],\quad \mathbf{x}(0)=\mathbf{x}_0,
\end{equation}
converges asymptotically to the left-hand Perron vector $\mathbf{c}\in\R^{N\times 1}$ of the matrix $P_A$.
\end{thm}

\begin{proof}
This case can be treated as a limiting case of weights of the form $\omega_{\varepsilon}(t)= \chi_{\raisebox{-.5ex}{$\scriptstyle [0,\varepsilon]$}}(t)$, and is related to a memory of exactly the “previous instant” of instant t. We want to solve the IVP for $t > 0$. For that particular  $t > 0$, we choose the family  $\omega_{\varepsilon}(t)= \chi_{\raisebox{-.5ex}{$\scriptstyle [0,\varepsilon]$}}(t)$ such that $0<\varepsilon<t$. Then IVP described in equation (\ref{eq:IVP-delta-original}) turns into
\begin{align}\label{eq:delta-case}
\dot{\mathbf{x}}^\top(t) & =\mathbf{x}^\top(t)(\lambda P_A-I_N)+(1-\lambda)\left[\frac{1}{\displaystyle\int_0^t \chi_{\raisebox{-.5ex}{$\scriptstyle [0,\varepsilon]$}}(u)\,du}\int_0^t \chi_{\raisebox{-.5ex}{$\scriptstyle [0,\varepsilon]$}}(t-u)\,\mathbf{x}^\top(u)\,du\right]\nonumber\\
& =\mathbf{x}^\top (t)(\lambda P_A-I_N)+(1-\lambda)\left[\frac{1}{\varepsilon}\int_{t-\varepsilon}^t \mathbf{x}^\top(u)\,du\right].
\end{align}
Now, observe that the term $\displaystyle\frac{1}{\varepsilon}\int_{t-\varepsilon}^t \mathbf{x}^\top(u)\,du$ represents the average PageRank-value over the interval $[t - \varepsilon, t]$. If we take the limit as $\varepsilon$ tends to $0$, equation (\ref{eq:delta-case}) into the following first-order linear differential equation with constant coefficients
\begin{equation*} 
\dot{\mathbf{x}}^\top(t) =\mathbf{x}^\top (t)(\lambda P_A-I_N)+(1-\lambda)\mathbf{x}^\top (t)= \mathbf{x}^\top (t)[-\lambda (I_N-P_A)].
\end{equation*}
For simplicity of notation, let $\mathbf{\Omega}=-\lambda (I_N-P_A)$. Therefore, the IVP described above has the form
\begin{equation}\label{eq:IVP-delta}
\dot{\mathbf{x}}^\top(t)=\mathbf{x}^\top (t)\,\mathbf{\Omega}(t),\quad  \mathbf{x}(0)=\mathbf{x}_0,  
\end{equation} 
whose solution has the form $\mathbf{x}^\top(t)=\mathbf{x}_0^\top e^{\,\mathbf{\Omega}\,t}$, for all $t>0$. In order to compute the matrix $e^{\mathbf{\,\Omega}\,t}$ we use the spectral resolution of matrix functions, in particular for systems of differential equations as described in \cite[Example 7.9.6]{Meyer}. More precisely, the solution of IVP in equation (\ref{eq:IVP-delta}) has the form
\begin{equation}\label{eq:solution-unperturbed}
\mathbf{x}^\top(t)=\mathbf{x}_0^\top e^{\,\mathbf{\Omega}\,t} = \sum_{i=1}^s\,\sum_{j=0}^{k_i-1}\,\frac{t^j\,e^{\,\alpha_i\,t}}{j!}\,\mathbf{v}_j^\top(\alpha_i),\quad \text{where}\quad \mathbf{v}_j^\top(\alpha_i)=\mathbf{Y}_0^\top\,\mathbf{G}_i\,(\,\mathbf{\Omega}-\alpha_i I_N)^{j},
\end{equation}
the set $\sigma(\mathbf{\Omega})=\{\alpha_1,\alpha_2,\dots,\alpha_s\}$ are the eigenvalues of the matrix $\mathbf{\Omega}=-\lambda (I_N-P_A)$ with corresponding indexes $\text{ind}(\lambda_i)=k_i$ and $\mathbf{G}_i$ are the spectral projectors associated to the eigenvalue $\alpha_i$ (that is, $\mathbf{G}_i$ is the projector onto the generalized eigenspace $\ker(\,(\mathbf{\Omega}-\alpha_i I_N)^{k_i}\,)$ along the range $R(\,(\mathbf{\Omega}-\alpha_i I_N)^{k_i}\,)$ ). Therefore, since the solution is given in terms of the eigenvalues of matrix $\mathbf{\Omega}=-\lambda (I_N-P_A)$ and its corresponding indexes and spectral projectors, we need to compute the set $\sigma(\mathbf{\Omega})$. In fact we will show that the eigenvalues of $\mathbf{\Omega}$ are closely related to the eigenvalues of the row-normalization matrix $P_A$.\\

From the definition of an eigenvalue and eigenvector of the matrix $\mathbf{\Omega}=-\lambda (I_N-P_A)$, we look for a non-zero vector $\mathbf{v}\in\R^{N\times1}$ and a complex number $\alpha\in\C$ such that $\mathbf{v}^\top\mathbf{\Omega}=\alpha\mathbf{v}^\top$, that is, $\mathbf{v}^\top\,\left[-\lambda (I_N-P_A)\right]=\alpha\mathbf{v}^\top$. Now, since the damping factor $\lambda\neq0$, a straightforward computation shows that the vector $\mathbf{v}^\top\in\R^{N\times 1}$ must satisfy the equality $ \mathbf{v}^\top\,P_A=\displaystyle\frac{\alpha+\lambda}{\lambda}\mathbf{v}^\top$. Therefore, the vector is an eigenvector of $P_A$ associated to the eigenvalue $\displaystyle\beta=\frac{\alpha+\lambda}{\lambda}$. Now, observe that if $\mu\in\C$ is an eigenvalue of $P_A$, we can take $\alpha=\lambda(\mu-1)$  and thus  
$\displaystyle\frac{\alpha+\lambda}{\lambda}=\mu$. Consequently all $\C$-eigenvalues $\alpha$ of $\mathbf{\Omega}=-\lambda (I_N-P_A)$ have the form $\alpha=\lambda(\mu-1)$, where $\lambda\in(0,1)$ is the damping factor and $\mu\in\C$ is an eigenvalue of $P_A$.\\

At this point, by the Perron-Frobenius Theorem (see \cite[Section 8.3]{Meyer}) applied to the irreducible and non-negative matrix $P_A$, there exists a unique positive vector $\mathbf{c}\in\mathbb{R}^{N\times 1}$ with $\mathbf{c}^\top\mathbf{e}=1$ (the left-hand Perron vector) associated to the eigenvalue $\rho(P_A)=1$, such that $\mathbf{c}^TP_A=\mathbf{c}^T$. For the special case of $\mu=1$, the dominant eigenvalue of the matrix $P_A$, the corresponding eigenvalue for the matrix $\mathbf{\Omega}=-\lambda (I_N-P_A)$ is $\alpha_1=0$. Moreover, by the linear relationship between the eigenvalues of $P_A$ and those of $\mathbf{\Omega}$, the eigenvalue $\alpha_1=0$ is simple. Notice also that for  the remaining eigenvalues $\mu$ of $P_A$ with modulus strictly less than $1$, the corresponding eigenvalues $\alpha$ of $\mathbf{\Omega}$ satisfy $\Re(\alpha) < 0$, where $\Re(\alpha)$ stands for the real part of  $\alpha$. Therefore, we can apply the ideas from \cite[Example 7.9.6]{Meyer} to obtain the asymptotic behavior
\begin{equation}\label{eq:asymptotic-exponential-delta}
t^j e^{\,\alpha\, t} \xrightarrow[t \to \infty]{}
\begin{cases}
0 & \text{if } \Re(\alpha_i) < 0,\text{ with }\alpha_i\neq0,\\
1 & \text{if } \alpha_i = 0 \text{ and } j = 0.
\end{cases}
\end{equation}
Therefore, substituting the asymptotic behavior form equation (\ref{eq:asymptotic-exponential-delta}) into the solution of the system given by equation (\ref{eq:IVP-delta}) we obtain
\begin{equation}\label{eq:asymptotic-solution-delta}
\lim_{t\to\infty} \mathbf{x}^\top(t)=\lim_{t\to\infty} \mathbf{x}_0^\top e^{\,\mathbf{\Omega}\,t} =\lim_{t\to\infty} \, \sum_{i=1}^s\,\sum_{j=0}^{k_i-1}\,\frac{t^j\,e^{\,\alpha_i\,t}}{j!}\,\mathbf{v}_j^\top(\alpha_i)=\mathbf{x}_0^\top \mathbf{v}_0^\top(\alpha_1=0)=\mathbf{x}_0^\top\mathbf{G}_1,  
\end{equation}
where $\mathbf{G}_1$ is the spectral projector associated to the eigenvalue $\alpha_1=0$ of the matrix $\mathbf{\Omega}=-\lambda (I_N-P_A)$. Observe that, since the eigenvalue $\alpha_1 = 0$ is simple, it follows from the result on spectral projectors associated with simple eigenvalues (see \cite[Chapter 7, Section 7.2]{Meyer}) that $\mathbf{G}_1=\mathbf{u}\mathbf{w}^\top/\mathbf{w}^\top\mathbf{u}$, where $\mathbf{u}$ and $\mathbf{w}^\top$ are the right-hand and the left-hand eigenvectors, respectively, associated to the simple eigenvalue $\alpha_1=0$ of the matrix $\mathbf{\Omega}$. Equivalently, since $\alpha_1=0$, the vectors $\mathbf{u}$ and $\mathbf{w}^\top$ must satisfy $\mathbf{\Omega}\,\mathbf{u}=\mathbf{0}_{N\times1}$ and $\mathbf{w}^\top\,\mathbf{\Omega}=\mathbf{0}^\top_{N\times1}$, that is, $\lambda (I_N-P_A)\,\mathbf{u}=\mathbf{0}_{N\times1}$ and $\mathbf{w}^\top\,\lambda (I_N-P_A)=\mathbf{0}^\top_{N\times1}$. Since $P_A$ is  row-stochastic, a straightforward computation shows $\mathbf{u}=\mathbf{e} = (1,\dots,1)^\top$ and $\mathbf{w}^\top=\mathbf{c}^\top$, where $\mathbf{c}\in\R^{N\times 1}$ is the left-hand Perron vector of the matrix $P_A$. Therefore, the spectral projector $\mathbf{G}_1$associated with the eigenvalue $\alpha_1=0$ is given by $\displaystyle\mathbf{G}_1=\frac{\mathbf{u}\mathbf{w}^\top}{\mathbf{w}^\top\mathbf{u}}=\frac{\mathbf{e}\,\mathbf{c}^\top}{\mathbf{c}^\top\mathbf{e}}=\mathbf{e}\,\mathbf{c}^\top$, where we have used the normalization condition $\mathbf{c}^\top\mathbf{e}=1$. Finally, using the expression of the spectral projector $\mathbf{G}_1$ in equation (\ref{eq:asymptotic-solution-delta}) together with the identity $\mathbf{x}_0^\top\mathbf{e}=1$, we obtain
\begin{equation*}
 \lim_{t\to\infty} \mathbf{x}^\top(t)=\lim_{t\to\infty} \mathbf{x}_0^\top e^{\,\mathbf{\Omega}\,t} =\mathbf{x}_0^\top\mathbf{G}_1=\mathbf{x}_0^\top\mathbf{e}\,\mathbf{c}^\top=\mathbf{c}^\top.
\end{equation*}
This completes the proof.\\
\end{proof}


\subsection[]{The exponential decay function $\omega(t)=e^{-at}$, for $t\geq0$, with parameter $a>0$.}\label{subsect:asymptotic-exponential}

In this subsection, we consider the weight function $\omega:[0,\infty)\to[0,\infty)$ to be an exponential decay type function, that is, $\omega(t)=e^{-at}$ for $t\geq0$, where $a>0$ is a parameter. This choice corresponds to the case in which the system exhibits an exponentially fading memory effect. Roughly speaking, past values of the vector $\mathbf{x}(s)$ associated with times $s$ closer to the current time $t$ have a greater influence on the dynamics than values corresponding to earlier times, particularly those near the initial instant $t=0$. Therefore, the contribution of the past states decreases exponentially as the temporal distance between $s$ and $t$ increases. In this sense, the parameter $a>0$ controls the memory decay rate of the system: larger values of $a$ produce a faster loss of memory, whereas smaller values of $a$ allow past states to retain influence over a longer time interval. From the neuroscience perspective, this exponential memory function is related to the so-called ``forgetting curve''; a mathematical model introduced in 1885 by the psychologist Hermann Ebbinghaus to describe how information retention decreases over time in the absence of reinforcement or repeated exposure. We refer the reader to \cite{MurreDros} for recent results concerning the replication of Ebbinghaus's experiments.\\

Observe that, for the exponential memory function $\omega(t)=e^{-at}$ for $t\geq0$, the corresponding integral can be evaluated in closed form. Therefore, in this case the IVP described in equation (\ref{eq:IVP-asymptotic}) turns into
\begin{align}\label{eq:exponential-case-numerical}
\dot{\mathbf{x}}^\top(t)=\mathbf{x}^\top(t)(\lambda P_A-I_N)+(1-\lambda)\left[\frac{a}{1-e^{-at}}\int_0^t e^{-a(t-u)}\,\mathbf{x}^\top (u)\,du\right],\quad \mathbf{x}(0)=\mathbf{x}_0,
\end{align}
where, as above,  $P_A$ is the row-normalization of the adjacency matrix $A$ associated with the directed graph $\mathcal{G}$, $\lambda\in(0,1)$ is the damping factor, and the probability vector $\mathbf{x}_0\in\mathbb{R}^{N\times1}$ is the initial condition.\\

Similarly to the previous subsection, we will start by performing some numerical simulations to examine the asymptotic behavior of the solution for the IVP in equation (\ref{eq:exponential-case-numerical}) when  $\omega(t)=e^{-at}$, for $t\geq0$. For numerical computations, we will take $a=1$.\\

We begin by considering the strongly connected directed 3-cycle graph shown in Figure \ref{fig:numerical-expo1}(a). To illustrate the behavior of the components of the solution $\mathbf{x}(t)=(\mathbf{x}_1(t),\mathbf{x}_2(t),\mathbf{x}_3(t))$ of the IVP in equation \eqref{eq:exponential-case}, we compute the corresponding numerical solution. In Figure \ref{fig:numerical-expo1}(b) we consider the damping factor $\lambda=0.5$ and the initial condition $\mathbf{x}_0^\top=(0.25,0.5,0.25)$, while in Figure \ref{fig:numerical-expo1}(c) we consider $\lambda=0.85$ and $\mathbf{x}_0^\top=(0.1,0.3,0.6)$.
\begin{figure}[H]
        \centering
        \includegraphics[width=1\linewidth]{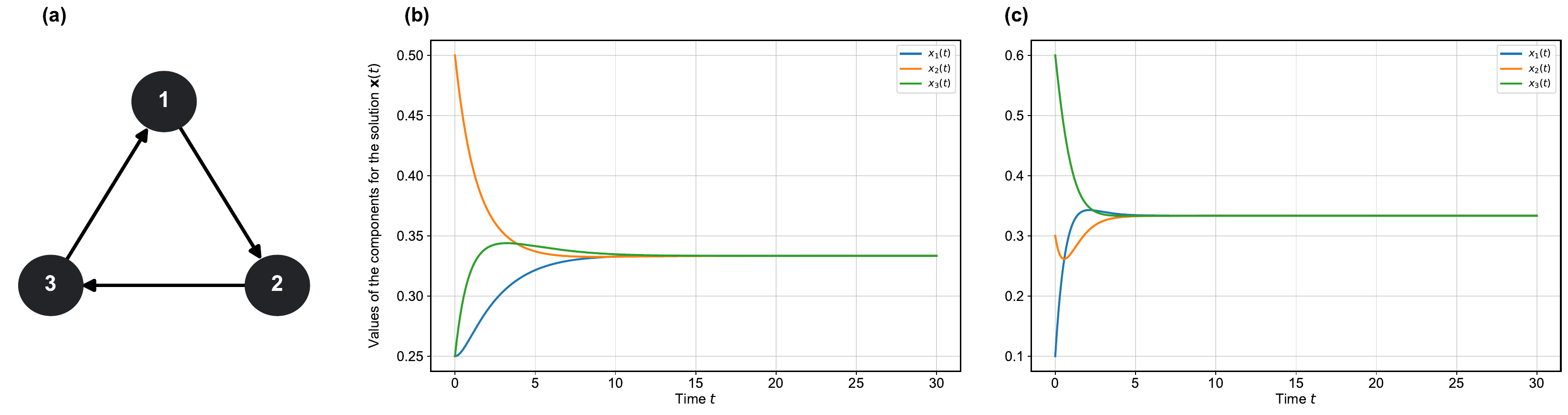}
        \caption{A numerical solution for the IVP in equation (\ref{eq:exponential-case-numerical}) for the 3-cycle digraph and $\omega(t)=e^{-t}$, for $t\geq0$.}
        \label{fig:numerical-expo1}
    \end{figure}
The numerical simulations shown in Figure \ref{fig:numerical-expo1} suggest that the solution $\mathbf{x}(t)$ is asymptotically stable and converges to the left-hand Perron vector $(1/3,1/3,1/3)^\top\in\R^{3\times1}$ of the row-normalization matrix $P_A$ associated with the adjacency matrix $A$ of the 3-cycle digraph in Figure \ref{fig:numerical-expo1}(a). \\

Now, to evaluate the robustness of this phenomenon with respect to changes in the graph topology, in Figure \ref{fig:numerical-expo2}(a) we consider a digraph which has bidirectional edges between nodes $1$-$2$ and $2$–$3$. To illustrate the behavior of the components of the solution $\mathbf{x}(t)=(\mathbf{x}_1(t),\mathbf{x}_2(t),\mathbf{x}_3(t))$ of the IVP in equation \eqref{eq:exponential-case-numerical}, we compute the corresponding numerical solution. As above, in Figure \ref{fig:numerical-expo2}(b) we consider the damping factor $\lambda=0.5$ and the initial condition $\mathbf{x}_0^\top=(0.2,0.1,0.7)$, while in Figure \ref{fig:numerical-expo2}(c) we consider $\lambda=0.85$ and $\mathbf{x}_0^\top=(0.1,0.3,0.6)$.
\begin{figure}[H]
        \centering
        \includegraphics[width=1\linewidth]{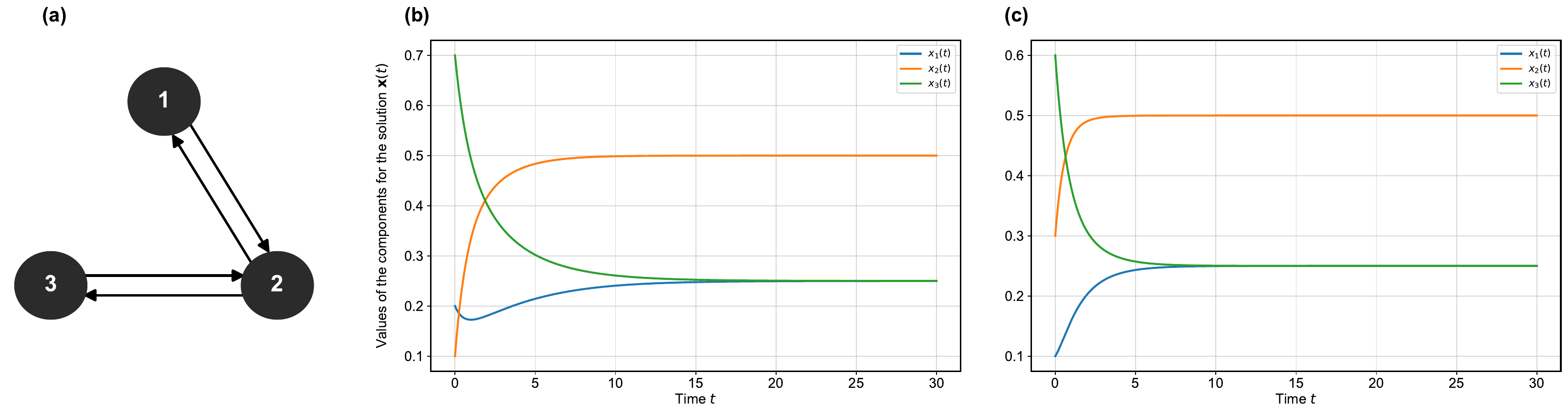}
        \caption{A numerical solution for the IVP in equation (\ref{eq:exponential-case-numerical}) for the bidirectional digraph and $\omega(t)=e^{-t}$, for $t\geq0$.}
        \label{fig:numerical-expo2}
    \end{figure}
    
Similarly to the $3$-cycle digraph case, the numerical simulations shown in Figure \ref{fig:numerical-expo2} suggest that the solution $\mathbf{x}(t)$ is asymptotically stable and converges to  the left-hand Perron vector $(0.25,0.50,0.25)^\top\in\R^{3\times1}$ of the row-normalization matrix $P_A$ associated with the adjacency matrix $A$ for the digraph in Figure \ref{fig:numerical-expo2}(a). \\

In view of the  simulations above  
we conjecture that the solution of the IVP asymptotically stabilizes at the left-hand Perron vector of the matrix $P_A$, when the exponential decay memory function is considered. This statement is formalized in Theorem \ref{thm:exponential-case}. In contrast to the Dirac delta weight function considered in Theorem \ref{thm:delta-case}, in the exponentially decreasing case we cannot derive an explicit closed-form expression for the solution for the IVP described in equation (\ref{eq:IVP-asymptotic}). This difficulty arises from the fact that the resulting equation is a linear integro-differential convolution equation, whose explicit solutions are generally hard to compute and are only available in certain particular cases (see \cite[Chapter 3]{GLS}, for instance). However, in spite of this limitation for providing an explicit formula for the solution of the IVP, we can study the asymptotic behavior of the solution using the ideas of asymptotic integration of linear differential systems described in \cite{HarrisLutz}. In our notation, it is well-known that the asymptotic behavior of the solution  of the linear system defined as $\dot{\mathbf{x}}^\top(t)=\mathbf{x}^\top(t)\,\mathbf{\Omega}$ is completely determined by the initial the spectrum of the constant coefficient matrix $\mathbf{\Omega}$. More precisely, the location of the eigenvalues of the matrix $\mathbf{\Omega}$ in the complex plane governs the growth, decay, and stability properties of the corresponding trajectories $\mathbf{x}(t)$ as $t$ grows to $\infty$. In particular, spectral properties such as the sign of the real parts of the eigenvalues determine whether solutions decay exponentially, remain bounded, or grow without bound.\\

A fundamental fact in perturbation theory is that these qualitative asymptotic properties are preserved under sufficiently small perturbations. Specifically, if $\mathbf{B}(t)$ is a sufficiently small time-dependent perturbation, then the solutions of perturbed linear system $\dot{\mathbf{x}}^\top(t)=\mathbf{x}^\top(t)\,\left[\mathbf{\Omega}+\mathbf{B}(t)\right]$ exhibits the same asymptotic and stability behavior as those for the unperturbed system $\dot{\mathbf{x}}^\top(t)=\mathbf{x}^\top(t)\,\mathbf{\Omega}$. In other words, under appropriate assumptions on the perturbation matrix $\mathbf{B}(t)$, the qualitative long-time dynamics of the system remain unchanged, and thus the dominant behavior of the solutions continues to be dictated by the spectrum of $\mathbf{\Omega}$. As an example, suppose that the perturbation $\mathbf{B}(t)$ is a continuous matrix-valued function for $t\geq t_0$, for some $t_0\geq0$, satisfying the growth condition $\int_{t_0}^\infty \vertiii{\,\mathbf{B}(t)\,}_\infty\,dt<\infty$,  where $\vertiii{\,\cdot\,}_\infty$ denotes the matrix norm in Remark \ref{rmk:operator-norm} of Section \ref{Section:Main results}. Then, the Levinson's classical result  asserts that the perturbed linear system $\dot{\mathbf{x}}^\top(t)=\mathbf{x}^\top(t)\,\left[\mathbf{\Omega}+\mathbf{B}(t)\right]$ admits a solution $\mathbf{x}(t)$ satisfying $\mathbf{x}^\top(t)=\left[\mathbf{C}^\top+o(1)\right]e^{\mathbf{\Omega}\, t}$, as $t\to\infty$, where $\mathbf{C}\in\R^{N\times1}$ is a constant matrix (see \cite{Levinson} or \cite{HarrisLutz}, for instance). Here, $o(1)$ denotes a row-matrix that converges to zero, as $t\to \infty$. Roughly speaking, the solutions of the perturbed system $\dot{\mathbf{x}}^\top(t)=\mathbf{x}^\top(t)\,\left[\mathbf{\Omega}+\mathbf{B}(t)\right]$ behave asymptotically like those of the unperturbed system $\dot{\mathbf{x}}^\top(t)=\mathbf{x}^\top(t)\,\mathbf{\Omega}$, whenever the perturbation $\mathbf{B}(t)$ is integrable.\\

With Levinson's classical result at hand, we are now in a position to prove that the solution of the PageRank dynamical system associated with the weight function $\omega(t)=e^{-at}$, for $t\geq0$ and $a>0$ as a parameter, converges asymptotically to the left-hand Perron vector of the row-normalized matrix $P_A$. For simplicity of notation, in the proof of the following result we consider only the case $a=1$. The extension to arbitrary $a>0$ follows by analogous arguments.\\

\begin{thm}\label{thm:exponential-case}
Let $A$ be a non-negative irreducible square matrix of order $N$ of a directed graph $\mathcal{G}$ and let $P_A$ be its row-normalization described in Section \ref{Section:Notation}. Let $\omega(t)=e^{-t}$, for $t\geq0$, be the non-negative and integrable function weight function. Then, for any personalization vector 
$\mathbf{x}_0\in\R^{N\times1}$ with $\mathbf{x}_0\geq0$ and $\mathbf{x}_0^\top\mathbf{e}=1$ and for any $\lambda\in(0,1)$, the solution $\mathbf{x}$ of the Initial Value Problem (IVP) given by
\begin{equation}\label{eq:IVP-exponential}
\dot{\mathbf{x}}^\top(t)=\mathbf{x}^\top (t)(\lambda P_A-I_N)+(1-\lambda)\left[\frac{1}{\displaystyle\int_0^t \omega(u)\,du}\int_0^t \omega(t-u)\,\mathbf{x}^\top(u)\,du\right],\quad \mathbf{x}(0)=\mathbf{x}_0,
\end{equation}
converges asymptotically to the left-hand Perron vector $\mathbf{c}\in\R^{N\times 1}$ of the matrix $P_A$.
\end{thm}

\begin{proof}
In this case, since the weight function $\omega(t)=e^{-t}$ for $t\geq0$, we have
\begin{equation*}
\int_0^t \omega(u)\,du=\int_0^t e^{-u}\,du=[-e^{-u}]_0^t=1-e^{-t},\quad (t\geq0).    
\end{equation*}
Therefore, the IVP described in equation above (\ref{eq:IVP-exponential}) can be rewritten as
\begin{equation}\label{eq:exponential-case}
\dot{\mathbf{x}}^\top(t)=\mathbf{x}^\top(t)(\lambda P_A-I_N)+(1-\lambda)\left[\frac{1}{1-e^{-t}}\int_0^t e^{-(t-u)}\,\mathbf{x}^\top (u)\,du\right],\quad \mathbf{x}(0)=\mathbf{x}_0.
\end{equation}
Multiplying the integro-differential equation in (\ref{eq:exponential-case}) by the factor $(1-e^{-t})$ we obtain the equivalent equation 
\begin{equation}\label{eq:exponential-multiplying}
(1-e^{-t})\dot{\mathbf{x}}^\top(t)=\mathbf{x}^\top(t)(1-e^{-t})(\lambda P_A-I_N)+(1-\lambda)\int_0^t e^{-(t-u)}\,\mathbf{x}^\top (u)\,du.
\end{equation}
Since this an  integro-differential equation, we take  the derivative and thus we obtain the following second order differential equation 

\begin{equation}\label{eq:derivation-exponential}
(e^{-t})\,\dot{\mathbf{x}}^\top(t)+(1-e^{-t})\,\ddot{\mathbf{x}}^\top(t)=\left[(e^{-t})\,\mathbf{x}^\top(t)+(1-e^{-t})\,\dot{\mathbf{x}}^\top(t)\right](\lambda P_A-I_N)-(1-\lambda)\int_0^t e^{-(t-u)}\,\mathbf{x}^\top (u)\,du,
\end{equation}

where in the integral term we have used the Leibniz integral rule for the differentiation under the integral sign. Now, writing  equation (\ref{eq:exponential-multiplying}) as 
\begin{equation*}
(1-\lambda)\int_0^t e^{-(t-u)}\,\mathbf{x}^\top (u)\,du=(1-e^{-t})\,\dot{\mathbf{x}}^\top(t)-\mathbf{x}^\top(t)(1-e^{-t})(\lambda P_A-I_N),
\end{equation*}

and using this expression in 
(\ref{eq:derivation-exponential}) 
we 
obtain the following second order differential equation 
\begin{align*}
(e^{-t})\,\dot{\mathbf{x}}^\top(t)+(1-e^{-t})\,\ddot{\mathbf{x}}^\top(t)
&=\left[(e^{-t})\,\mathbf{x}^\top(t)+(1-e^{-t})\,\dot{\mathbf{x}}^\top(t)\right](\lambda P_A-I_N)-(1-\lambda)\int_0^t e^{-(t-u)}\,\mathbf{x}^\top (u)\,du,\\
&=\left[(e^{-t})\,\mathbf{x}^\top(t)+(1-e^{-t})\,\dot{\mathbf{x}}^\top(t)\right](\lambda P_A-I_N)-\left[(1-e^{-t})\,\dot{\mathbf{x}}^\top(t)-\mathbf{x}^\top(t)(1-e^{-t})(\lambda P_A-I_N)\right].
\end{align*}
which, after rearranging and dividing by the factor $(1-e^{-t})$, becomes  the following second-order linear differential equation with variable coefficients 
\begin{equation}\label{eq:second-ODE}
\ddot{\mathbf{x}}^\top(t)+\dot{\mathbf{x}}^\top(t)\,\mathbf{a}(t)+\mathbf{x}^\top(t)\,\mathbf{b}(t)=0,
\end{equation}
where the time-dependent matrix coefficients $\mathbf{a}(t)$ and $\mathbf{b}(t)$ are defined as follows
\begin{equation*}
\mathbf{a}(t)=\frac{e^{-t}}{1-e^{-t}}I_N-(\lambda P_A-I_N) +I_N\quad ,\quad \mathbf{b}(t)=-(\lambda P_A-I_N)\frac{e^{-t}}{1-e^{-t}}-\frac{(1-\lambda)}{1-e^{-t}}I_N-(\lambda P_A-I_N).
\end{equation*}
Since the problem involves a second-order ordinary differential equation described by (\ref{eq:second-ODE}), two initial conditions must be specified. In this case, we consider:

\begin{itemize}
\item[(C1)] The initial condition $\mathbf{x}(0)=\mathbf{x}_0$ with $\mathbf{x}_0\geq0$ and $\mathbf{x}_0^\top\mathbf{e}=1$.
\item[(C2)] A condition on $\dot{\mathbf{x}}^\top(0)$. This can be derived by continuity and by using equation~(\ref{eq:exponential-case}), which defines the dynamical system. Specifically,
\begin{align*}
\dot{\mathbf{x}}^\top(0)=\lim_{t\to0^+}\dot{\mathbf{x}}^\top(t)
&=\lim_{t\to0^+}\left\{\mathbf{x}^\top(t)(\lambda P_A-I_N)+(1-\lambda)\left[\frac{1}{1-e^{-t}}\int_0^te^{-(t-u)}\mathbf{x}^\top(u)\,du\right]\right\}\\
&=\mathbf{x}^\top(0)(\lambda P_A-I_N)+(1-\lambda)\lim_{t\to0^+}\left[\frac{1}{1-e^{-t}}\int_0^t e^{-(t-u)}\mathbf{x}^\top(u)\,du\right]\\
&=\mathbf{x}^\top(0)(\lambda P_A-I_N)+(1-\lambda)\mathbf{x}^\top(0)\\
&=\mathbf{x}^\top(0)\left[(\lambda P_A-I_N) + (1-\lambda)I_N\right]=\mathbf{x}^\top(0)\lambda(P_A-I_N).
\end{align*}
Hence, $\dot{\mathbf{x}}^\top(0) = \mathbf{x}^\top(0)\lambda(P_A-I_N)$.
\end{itemize}

Now, the second-order linear differential equation (\ref{eq:second-ODE}) can be converted into a first-order system by introducing new variables, typically $\mathbf{y}_1^\top=\mathbf{x}^\top$ and $\mathbf{y}_2^\top=\dot{\mathbf{x}}^\top$. The first-order linear system can be written in the matrix form
\begin{equation}\label{eq:first-order-system}
\mathbf{\dot{Y}}^\top(t)=
(\enspace\dot{\mathbf{y}}_1^\top(t)\enspace|\enspace
\dot{\mathbf{y}}_2^\top(t)\enspace)=(\enspace\mathbf{y}_1^\top(t)\enspace|\enspace
\mathbf{y}_2^\top(t)\enspace)\,\mathbf{M}(t)=\mathbf{Y}^\top(t)\,\mathbf{M}(t), \quad (t>0),    
\end{equation}
where the coefficient-matrix $\mathbf{M}(t)$ is defined as follows
\begin{equation}\label{eq:matrix-system}
\mathbf{M}(t)=
\begin{pmatrix}
\mathbf{0}_{N\times N} & -\mathbf{b}(t)\\
I_N & -\mathbf{a}(t)
\end{pmatrix}=
\begin{pmatrix}
\mathbf{0}_{N\times N} & (\lambda P_A-I_N)\displaystyle\frac{e^{-t}}{1-e^{-t}}+\frac{(1-\lambda)}{1-e^{-t}}I+(\lambda P_A-I_N)\\
I_N&  -\displaystyle\frac{e^{-t}}{1-e^{-t}}I+(\lambda P_A-2I_N)
\end{pmatrix}.
\end{equation}
Observe that, since
\begin{equation*}
\frac{(1-\lambda)}{1-e^{-t}}I_N=(1-\lambda)\frac{1}{1-e^{-t}}I_N=(1-\lambda)\left[1+\frac{e^{-t}}{1-e^{-t}}\right]I_N=(1-\lambda)I_N + (1-\lambda)\frac{e^{-t}}{1-e^{-t}}I_N,\quad (t>0),
\end{equation*}
the coefficient $-\mathbf{b}(t)$ can be rewritten as 
\begin{equation*}
-\mathbf{b}(t)=(\lambda P_A-I_N)\frac{e^{-t}}{1-e^{-t}}+\frac{(1-\lambda)}{1-e^{-t}}I_N+(\lambda P_A-I_N)=\lambda(P_A-I_N)\frac{e^{-t}}{1-e^{-t}}+\lambda (P_A-I_N).
\end{equation*}
Therefore, matrix $\mathbf{M}(t)$ in equation (\ref{eq:first-order-system}) can be expressed as follows
\begin{equation}\label{eq:matrix-system}
\mathbf{M}(t)=
\begin{pmatrix}
\mathbf{0}_{N\times N} & -\mathbf{b}(t) \\
I_N & -\mathbf{a}(t)
\end{pmatrix}=
\begin{pmatrix}
\mathbf{0}_{N\times N} & \lambda(P_A-I_N)\displaystyle\frac{e^{-t}}{1-e^{-t}}+\lambda (P_A-I_N)\\
 I_N &  -\displaystyle\frac{e^{-t}}{1-e^{-t}}I_N+(\lambda P_A-2I_N)
\end{pmatrix}.
\end{equation}
At this point, the idea is to apply the asymptotic integration to the first-order linear system described in equation (\ref{eq:first-order-system}). To this aim, observe that the matrix $\mathbf{M}(t)$ can be decomposed as $\mathbf{M}(t)=\mathbf{\Omega}+\mathbf{B}(t)$, where
\begin{equation*}
\mathbf{\Omega}=
\begin{pmatrix}
\mathbf{0}_{N\times N} & \lambda(P_A-I_N)\\
I_N & \lambda P_A-2I_N
\end{pmatrix}\quad \text{ and }\quad 
\mathbf{B}(t)=
\begin{pmatrix}
\mathbf{0}_{N\times N} & \lambda(P_A-I_N)\displaystyle\frac{e^{-t}}{1-e^{-t}}\\
\mathbf{0}_{N\times N} & -\displaystyle\frac{e^{-t}}{1-e^{-t}}I_N
\end{pmatrix}.
\end{equation*}
Now, since for every $t_0>0$ we have 
\begin{equation*}
\int_{t_0}^\infty \vertiii{\,\mathbf{B}(t)\,}_\infty dt=
\bigint_{t_0}^\infty\vertiii{ \,\begin{pmatrix}
\mathbf{0}_{N\times N} & \lambda(P_A-I_N)\displaystyle\frac{e^{-t}}{1-e^{-t}}\\
\mathbf{0}_{N\times N} & -\displaystyle\frac{e^{-t}}{1-e^{-t}}I_N
\end{pmatrix}\,}_\infty dt\leq\left(\lambda\,\vertiii{\,P_A-I_N\,}_\infty+\vertiii{\,I_N\,}_\infty\right)\int_{t_0}^\infty \frac{e^{-t}}{1-e^{-t}}\,dt <\infty,
\end{equation*}

the perturbation matrix $\mathbf{B}(t)$ is integrable with respect to the $\infty$-norm matrix. As an application of the Levinson's theorem on the asymptotic integration (see \cite{Coddington-Levinson} or \cite{HarrisLutz}, for instance), the asymptotic behavior of solutions for the perturbed system $\dot{\mathbf{Y}}^\top(t)=\mathbf{Y}^\top(t)\left[\,\mathbf{\Omega}+\mathbf{B}(t)\,\right]$ are still determined by the limiting system $\dot{\mathbf{Y}}^\top(t)=\mathbf{Y}^\top(t)\,\mathbf{\Omega}$.
Therefore, we focus on the asymptotic behavior of the first-order linear system $\dot{\mathbf{Y}}^\top(t)=\mathbf{Y}^\top(t)\,\mathbf{\Omega}$. Since the coefficient matrix $\mathbf{\Omega}$ of the unperturbed linear systems is constant, it is clear that the solution of the first-order linear system $\dot{\mathbf{Y}}^\top(t)=\mathbf{Y}^\top(t)\,\mathbf{\Omega}$ has the form $\mathbf{Y}^\top(t)=\mathbf{Y}_0^\top e^{\,\mathbf{\Omega}\,t}$, where $\mathbf{Y}_0=\mathbf{Y}(0)$ is the initial value for the IVP. In our case, since $\mathbf{y}_1^\top=\mathbf{x}^\top$ and $\mathbf{y}_2^\top=\dot{\mathbf{x}}^\top$, we have $\mathbf{Y}_0^\top=\left(\enspace\mathbf{y}_1^\top(0)\enspace|\enspace\mathbf{y}_2^\top(0)\enspace\right)=\left(\enspace\mathbf{x}_0^\top\enspace|\enspace\mathbf{x}_0^\top\lambda(P_A-I_N)\enspace\right)$, which are the conditions (C1) and (C2) described above for the second-order ordinary differential equation.\\

Now, to compute the matrix $e^{\mathbf{\,\Omega}\,t}$ we use the spectral resolution of matrix functions, in particular for systems of differential equations described in \cite[Example 7.9.6]{Meyer}. More precisely, the solution of the unperturbed system $\dot{\mathbf{Y}}^\top(t)=\mathbf{Y}^\top(t)\,\mathbf{\Omega}$ has the form
\begin{equation}\label{eq:solution-unperturbed}
\mathbf{Y}^\top(t)=\mathbf{Y}_0^\top e^{\mathbf{\,\Omega}\,t} = \sum_{i=1}^s\,\sum_{j=0}^{k_i-1}\,\frac{t^j\,e^{\,\alpha_i\,t}}{j!}\,\mathbf{v}_j^\top(\alpha_i),\quad \text{where}\quad \mathbf{v}_j^\top(\alpha_i)=\mathbf{Y}_0^\top\,\mathbf{G}_i\,(\mathbf{\Omega}-\alpha_i I_N)^{j},
\end{equation}
the set $\sigma(\mathbf{\Omega})=\{\alpha_1,\alpha_2,\dots,\alpha_s\}$ are the eigenvalues of the matrix $\mathbf{\Omega}$ with corresponding indexes $\text{ind}(\lambda_i)=k_i$ and $\mathbf{G}_i$ are the spectral projectors associated to the eigenvalue $\alpha_i$ (that is, $\mathbf{G}_i$ is the projector onto the generalized eigenspace $\ker(\,(\mathbf{\Omega}-\alpha_i I_N)^{k_i}\,)$ along the range $R(\,(\mathbf{\Omega}-\alpha_i I_N)^{k_i}\,)$ ). Therefore, since the solution is in terms of the eigenvalues of matrix $\mathbf{\Omega}$ and its corresponding indexes and spectral projectors, we proceed to compute the set $\sigma(\mathbf{\Omega})$. In fact, as we show in a moment, the eigenvalues of $\mathbf{\Omega}$ are closely related to the eigenvalues of the row-normalization matrix $P_A$.\\

From the definition of an eigenvalue and eigenvector of the matrix $\mathbf{\Omega}$, we look for a non-zero vector $\mathbf{V}\in\R^{2N\times1}$ with $\mathbf{V}^\top=(\enspace\mathbf{V}_1^\top\enspace|\enspace\mathbf{V}_2^\top\enspace)$ and a number $\alpha\in\C$ such that $\mathbf{V}^\top\mathbf{\Omega}=\alpha\mathbf{V}^\top$, that is,
\begin{equation}\label{eq:eigenvalues-Omega}
(\enspace\mathbf{V}_1^\top\enspace|\enspace\mathbf{V}_2^\top\enspace)\begin{pmatrix}
\mathbf{0}_{N\times N} & \lambda(P_A-I_N)\\
I_{N} & \lambda P_A-2I_N
\end{pmatrix}=\alpha\,(\enspace\mathbf{V}_1^\top\enspace|\enspace\mathbf{V}_2^\top\enspace)
\end{equation}
From the equation (\ref{eq:eigenvalues-Omega}) we obtain the following system of equations
\begin{align}
\mathbf{V}_2^\top=\alpha\,\mathbf{V}_1^\top\label{eq:system-eigenvector1},\\
\mathbf{V}_1^\top\,\lambda(P_A-I_N)+\mathbf{V}_2^\top\,(\lambda P_A-2I_N)=\alpha\,\mathbf{V}_2^\top\label{eq:system-eigenvector2}.
\end{align}
Now, substituting (\ref{eq:system-eigenvector1}) in equation (\ref{eq:system-eigenvector2}) we obtain the equation in terms of $\mathbf{V}_1^\top$ as follows
\begin{equation}\label{eq:substituying-V_2}
\mathbf{V}_1^\top\,\lambda(P_A-I_N)+\alpha\mathbf{V}_1^\top\,(\lambda P_A-2I_N)=\alpha^2\,\mathbf{V}_1^\top
\end{equation}
A straightforward computation in equation (\ref{eq:substituying-V_2}) shows that the vector $\mathbf{V}_1^\top\in\R^{N\times 1}$ must satisfy the equation 
\begin{equation}\label{eq:condition-V_1}
\lambda(\alpha+1)\mathbf{V}_1^\top\,P_A=(\alpha^2+2\alpha+\lambda)\mathbf{V}_1^\top.    
\end{equation}
Now observe that $\alpha+1\neq0$. Indeed, if $\alpha+1=0$, by equation (\ref{eq:condition-V_1}) the vector $\mathbf{V}_1^\top$ must satisfy $\mathbf{0}_{N\times 1}^\top=(-1+\lambda)\mathbf{V}_1^\top$ and, since the damping factor $\lambda\in(0,1)$, it follows that $\mathbf{V}_1^\top=\mathbf{0}_{N\times 1}^\top$. Now, by equation (\ref{eq:system-eigenvector1}), we also obtain that $\mathbf{V}_2^\top=\mathbf{0}_{N\times 1}^\top$ and therefore $\mathbf{V}^\top=(\enspace\mathbf{V}_1^\top\enspace|\enspace\mathbf{V}_2^\top\enspace)=(\enspace\mathbf{0}_{N\times 1}^\top\enspace|\enspace\mathbf{0}_{N\times 1}^\top\enspace)$, which contradicts the definition of a non-zero vector.\\

Therefore, since $\lambda(\alpha+1)\neq0$ in equation (\ref{eq:condition-V_1}), the vector $\mathbf{V}_1\in\R^{N\times 1}$ satisfies the equation $\displaystyle\mathbf{V}_1^\top P_A=\beta\,\mathbf{V}_1^\top$ with $\displaystyle\beta=\frac{\alpha^2+2\alpha+\lambda}{\lambda(\alpha+1)}$. This shows that $\mathbf{V}_1^\top$ is an eigenvector of $P_A$ associated to the eigenvalue $\displaystyle\beta=\frac{\alpha^2+2\alpha+\lambda}{\lambda(\alpha+1)}$. Now, observe that if $\mu\in\C$ is an eigenvalue of $P_A$, then there exists $\alpha\in\C$ such that $\displaystyle\frac{\alpha^2+2\alpha+\lambda}{\lambda(\alpha+1)}=\mu$. Equivalently, the eigenvalues $\alpha\in\C$ of $\mathbf{\Omega}$ are the solution of the quadratic equation $\alpha^2+(2-\lambda\mu)\alpha+\lambda(1-\mu)=0$, where $\lambda\in(0,1)$ and $\mu\in\C$ is an eigenvalue of $P_A$. It is worth mentioning that for each eigenvalue $\mu$ of $P_A$, there correspond two eigenvalues $\alpha$ of the $\mathbf{\Omega}$.

At this point, by the Perron-Frobenius Theorem (see \cite[Section 8.3]{Meyer}) applied to the irreducible and non-negative matrix $P_A$, there exists a unique positive vector $\mathbf{c}\in\mathbb{R}^{N\times 1}$ with $\mathbf{c}^\top\mathbf{e}=1$ (the left-hand Perron vector) associated to the eigenvalue $\rho(P_A)=1$, such that $\mathbf{c}^TP_A=\mathbf{c}^T$. For the special case of $\mu=1$, the dominant eigenvalue of the matrix $P_A$, the eigenvalues for the matrix $\mathbf{\Omega}$ are the solutions of $\alpha^2+(2-\lambda)\alpha=0$, namely $\alpha_1=0$ and $\alpha_2=\lambda-2$. Moreover, the eigenvalue $\alpha_1=0$ is simple. Now, for the remaining eigenvalues $\mu$ of $P_A$ with modulus strictly less than $1$, observe that the coefficients $(2-\lambda\mu)$ and $\lambda(1-\mu)$ have positive real parts. Hence, an application of a complex-coefficient extension of the Routh–Hurwitz stability criterion (see \cite[Appendix F]{Yamamoto}, for instance) to the quadratic equation $\alpha^2+(2-\lambda\mu)\alpha+\lambda(1-\mu)=0$ shows that the real parts of eigenvalues $\alpha$ are strictly negative.

Now, since all eigenvalues $\alpha$ of $\mathbf{\Omega}$ satisfy that its real part $\Re(\alpha) < 0$ (except for the eigenvalue $\alpha_1 = 0$, associated with the Perron eigenvalue of the matrix $P_A$), we can apply the ideas from \cite[Example 7.9.6]{Meyer} to obtain the asymptotic behavior
\begin{equation}\label{eq:asymptotic-exponential}
t^j e^{\,\alpha_i\, t} \xrightarrow[t \to \infty]{}
\begin{cases}
0 & \text{if } \Re(\alpha_i) < 0,\text{ with }\alpha_i\neq0\\
1 & \text{if } \alpha_i = 0 \text{ and } j = 0.
\end{cases}
\end{equation}
Therefore, substituting the asymptotic behavior form equation (\ref{eq:asymptotic-exponential}) into the solution of the unperturbed system given in equation (\ref{eq:solution-unperturbed}) we obtain
\begin{equation}\label{eq:asymptotic-solution}
\lim_{t\to\infty} \mathbf{Y}^\top(t)=\lim_{t\to\infty} \mathbf{Y}_0^\top e^{\,\mathbf{\Omega}\,t} =\lim_{t\to\infty} \, \sum_{i=1}^s\,\sum_{j=0}^{k_i-1}\,\frac{t^j\,e^{\,\alpha_i\,t}}{j!}\,\mathbf{v}_j^\top(\alpha_i)=\mathbf{Y}_0^\top \mathbf{v}_0^\top(\alpha_1=0)=\mathbf{Y}_0^\top\mathbf{G}_1,  
\end{equation}
where $\mathbf{G}_1$ is the spectral projector associated to the eigenvalue $\alpha_1=0$ of the matrix $\mathbf{\Omega}$. Observe that, since the eigenvalue $\alpha_1 = 0$ is simple, it follows from the result on spectral projectors associated with simple eigenvalues (see \cite[Chapter 7, Section 7.2]{Meyer}) that $\mathbf{G}_1=\mathbf{u}\mathbf{w}^\top/\mathbf{w}^\top\mathbf{u}$, where $\mathbf{u}$ and $\mathbf{w}^\top$ are the right-hand and the left-hand eigenvectors, respectively, associated to the simple eigenvalue $\alpha_1=0$ of the matrix $\mathbf{\Omega}$. Now, we proceed with the computation of these right and left-hand eigenvectors of the matrix $\mathbf{\Omega}$ associated to the eigenvalue $\alpha_1=0$.\\

Firstly, for the left-hand eigenvector of $\mathbf{\Omega}$ we look for a vector $\mathbf{w}\in\R^{2N\times1}$ with $\mathbf{w}^\top=(\enspace\mathbf{w}_1^\top\enspace|\enspace\mathbf{w}_2^\top\enspace)$ such that 
\begin{equation}\label{eq:left-hand-eigenvector}
\mathbf{w}^\top\,\mathbf{\Omega}= 0\,\mathbf{w}^\top,\  \text{ that is},\ (\enspace\mathbf{w}_1^\top\enspace|\enspace\mathbf{w}_2^\top\enspace)\begin{pmatrix}
\mathbf{0}_{N\times N} & \lambda(P_A-I_N)\\
I_{N} & \lambda P_A-2I_N
\end{pmatrix}=\,(\enspace\mathbf{0}_{N\times 1}^\top\enspace|\enspace\mathbf{0}_{N\times 1}^\top\enspace).
\end{equation}
From equation (\ref{eq:left-hand-eigenvector}) we obtain the following system
\begin{align}
\mathbf{w}_2^\top=\mathbf{0}_{N\times 1}^\top,\label{eq:left-hand-eigenvector-system1}\\
\mathbf{w}_1^\top\,\lambda(P_A-I_N)+\mathbf{w}_2^\top\,(\lambda P_A-2I_N)=\mathbf{0}_{N\times1}^\top.\label{eq:left-hand-eigenvector-system2}
\end{align}
Obviously, equation (\ref{eq:left-hand-eigenvector-system1}) implies that $\mathbf{w}_2^\top=\mathbf{0}_{N\times 1}^\top$. Now, substituting equation (\ref{eq:left-hand-eigenvector-system1}) in equation (\ref{eq:left-hand-eigenvector-system2}) we get that $\mathbf{w}_1^\top\,\lambda(P_A-I_N)=\mathbf{0}_{N\times1}^\top$ and, since $\lambda\in(0,1)$, it follows that $\mathbf{w}_1^\top$ must satisfy $\mathbf{w}_1^\top\,P_A=\mathbf{w}_1^\top$. This means that the vector $\mathbf{w}_1\in\R^{N\times1}$ is an left-hand eigenvector of $P_A$ associated to the eigenvalue $\mu=1$. By the existence and uniqueness of the left-hand Perron vector $\mathbf{c}$ for the matrix $P_A$, we conclude that  $\mathbf{w}_1^\top=\mathbf{c}^\top$. Therefore, $\mathbf{w}^\top=(\enspace\mathbf{w}_1^\top\enspace|\enspace\mathbf{w}_2^\top\enspace)=(\enspace\mathbf{c}^\top\enspace|\enspace\mathbf{0}_{N\times1}^\top\enspace)$.\\

Now, we proceed to compute the right-hand eigenvector for the matrix $\mathbf{\Omega}$ associated to the eigenvalue $\alpha_1=0$. To this aim, we look for a column vector $\mathbf{u}\in\R^{2N\times1}$ with  $\mathbf{u}^\top=(\enspace\mathbf{u}_1^\top\enspace|\enspace\mathbf{u}_2^\top\enspace)$ such that
\begin{equation}\label{eq:right-hand-eigenvector}
\mathbf{\Omega}\mathbf{u}= 0\,\mathbf{u}\quad \text{, that is,}\quad\begin{pmatrix}
\mathbf{0}_{N\times N} & \lambda(P_A-I_N)\\
I_{N} & \lambda P_A-2I_N
\end{pmatrix}\,\begin{pmatrix}
\mathbf{u}_1\\
\mathbf{u}_2
\end{pmatrix}=\begin{pmatrix}
\mathbf{0}_{N\times1}\\
\mathbf{0}_{N\times1}
\end{pmatrix}.
\end{equation}
From equation (\ref{eq:right-hand-eigenvector}) we obtain the following system
\begin{align}
\lambda(P_A-I_N)\mathbf{u}_2=\mathbf{0}_{N\times 1},\label{eq:right-hand-eigenvector-system1}\\
\mathbf{u}_1+(\lambda P_A-2I_N)\mathbf{u}_2=\mathbf{0}_{N\times1}.\label{eq:right-hand-eigenvector-system2}
\end{align}
From (\ref{eq:right-hand-eigenvector-system1}), since $\lambda\in(0,1)$, it follows that $P_A\mathbf{u}_2=\mathbf{u}_2$, which implies that $\mathbf{u}_2=\mathbf{e}$. Now, substituting the value for the vector $\mathbf{u}_2$ in equation (\ref{eq:right-hand-eigenvector-system2}) we conclude that $\mathbf{u}_1=(2I_N-\lambda P_A)\mathbf{u}_2=(\lambda-2)\mathbf{e}$. Therefore, $\mathbf{u}^\top=(\enspace\mathbf{u}_1^\top\enspace|\enspace\mathbf{u}_2^\top\enspace)=(\enspace (\lambda-2)\,\mathbf{e}^\top\enspace|\enspace\mathbf{e}^\top\enspace)$.\\

Hence, using the right-hand and left-hand eigenvectors previously computed, we show that the spectral projector $\mathbf{G}_1$ has the form
\begin{equation}\label{eq:spectral-projector}
\mathbf{G}_1=\frac{\mathbf{u}\mathbf{w}^\top}{\mathbf{w}^\top\mathbf{u}}=\frac{
\begin{pmatrix}
(\lambda-2)\,\mathbf{e}\\
\mathbf{e}
\end{pmatrix}
(\enspace\mathbf{c}^\top\enspace|\enspace\mathbf{0}_{N\times1}^\top)}{(\enspace\mathbf{c}^\top\enspace|\enspace\mathbf{0}_{N\times1}^\top)\begin{pmatrix}
(\lambda-2)\,\mathbf{e}\\
\mathbf{e}
\end{pmatrix}}
=\frac{1}{\lambda-2}\begin{pmatrix}
(\lambda-2)\,\mathbf{e}\mathbf{c}^\top & \mathbf{0}_{N\times1}^\top\\
\mathbf{e}\mathbf{c}^\top & \mathbf{0}_{N\times1}^\top
\end{pmatrix},
\end{equation}
where, in the denominator, we have used the equality $\mathbf{c}^\top\mathbf{e}=1$ . Therefore, using the expression of the spectral projector in equation (\ref{eq:asymptotic-solution}) and the fact that $P_A\,\mathbf{e}=\mathbf{e}$ we get
\begin{align*}
\lim_{t\to\infty} \mathbf{Y}^\top(t)=\mathbf{Y}_0^\top\mathbf{G}_1
&=\frac{1}{\lambda-2}\left(\enspace\mathbf{x}_0^\top\enspace|\enspace\mathbf{x}_0^\top\lambda(P_A-I_N)\enspace\right)\begin{pmatrix}
(\lambda-2)\,\mathbf{e}\mathbf{c}^\top & \mathbf{0}_{N\times1}^\top\\
\mathbf{e}\mathbf{c}^\top & \mathbf{0}_{N\times1}^\top
\end{pmatrix}\\
&=\frac{1}{\lambda-2}\left(\enspace\mathbf{x}_0^\top(\lambda-2)\,\mathbf{e}\mathbf{c}^\top+\mathbf{x}_0^\top\lambda(P_A-I_N)\,\mathbf{e}\mathbf{c}^\top\enspace|\enspace\mathbf{0}_{N\times1}^\top\enspace\right)\\
&=\frac{1}{\lambda-2}\left(\enspace(\lambda-2)\,\mathbf{x}_0^\top\mathbf{e}\mathbf{c}^\top+\mathbf{0}_{N\times1}^\top\enspace|\enspace\mathbf{0}_{N\times1}^\top\enspace\right)\\
&=\frac{1}{\lambda-2}\left(\enspace(\lambda-2)\,\mathbf{c}^\top\enspace|\enspace\mathbf{0}_{N\times1}^\top\enspace\right)=\left(\enspace\mathbf{c}^\top\enspace|\enspace\mathbf{0}_{n\times1}^\top\enspace\right).
\end{align*}
Finally, since $\mathbf{Y}^\top(t) = \big(\enspace\mathbf{y}_1^\top(t)\enspace|\enspace\mathbf{y}_2^\top(t) \enspace)$, with $\mathbf{y}_1^\top= \mathbf{x}^\top$ and $\mathbf{y}_2^\top = \dot{\mathbf{x}}^\top$, for $t\geq0$, it follows that $\displaystyle\lim_{t\to\infty}\mathbf{x}^\top(t)=\mathbf{c}^\top$ where $\mathbf{c}$ is the left-hand Perron vector for the matrix $P_A$. This completes the proof.

\end{proof}

\subsection[]{The oscillatory memory function $\omega(t)=e^{at}\cos^2(bt)$, for $t\geq0$, with parameters $a,b>0$.}\label{subsect:asymptotic-oscillatory}

Finally, in this subsection we consider as the weight function $\omega:[0,\infty)\to[0,\infty)$ an exponential oscillatory  function of the form $\omega(t)=e^{at}\cos^2(bt)$, for $t\geq0$, where $a,b\in\R$ are positive parameters. Notice that this memory function $\omega$ exhibits two principal features:
\begin{itemize}
    \item In contrast to the previous example, the exponential factor $e^{at}$ with $a>0$, modulates the weighting of past states, assigning increasing (or decreasing, depending on the time-reversal convention in the) influence across the memory horizon.
    \item The trigonometric factor $\cos^2(bt)$ induces a periodic modulation of the memory weight. Moreover, whenever $\cos^2(bt)=0$, the memory contribution vanishes (no memory effect for these time instants), leading to what may be interpreted as intermittent ``memory gaps” in the process.
\end{itemize}

In this case the IVP described in equation (\ref{eq:IVP-asymptotic}) turns into
\begin{align}\label{eq:oscillatoy-case-numerical}
\dot{\mathbf{x}}^\top(t)=\mathbf{x}^\top(t)(\lambda P_A-I_N)+(1-\lambda)\left[\frac{1}{\displaystyle\int_0^te^{au}\cos^2(bu)\,du}\int_0^t e^{-a(t-u)}\cos^2(\,b(t-u)\,)\,\mathbf{x}^\top (u)\,du\right],\quad \mathbf{x}(0)=\mathbf{x}_0,
\end{align}
where $P_A$, $\lambda\in(0,1)$  and the initial condition $\mathbf{x}_0\in\mathbb{R}^{N\times1}$ are as above. \\


The identity $\cos^2(bu)=(\,1+\cos(2bu)\,)/2$, allows for the integral of the memory function $\omega(t)=e^{at}\cos^2(bt)$ to be expressed in a closed form as
\begin{equation}\label{eq:closed-formula-oscillatoy}
\int_0^t\omega(u)\,du=\int_0^te^{au}\cos^2(bu)\,du= \frac{e^{at}}{2} \left[ \frac{1}{a} + \frac{a\cos(2bt) + 2b\sin(2bt)}{a^2 + 4b^2} \right] - \frac{1}{2}\left( \frac{1}{a} + \frac{a}{a^2 + 4b^2} \right). 
\end{equation}
This formula will be very useful in the proof of Theorem \ref{thm:oscillatory-case}, where we study the asymptotic behavior of the solution for the IVP described in equation (\ref{eq:oscillatoy-case-numerical}).\\

Similarly to the previous subsection, we perform some numerical simulations to examine the asymptotic behavior of the solution for the IVP in equation (\ref{eq:oscillatoy-case-numerical}) when we consider the exponential oscillatory memory function $\omega(t)=e^{at}\cos^2(bt)$, for $t\geq0$. For numerical computations, we take $a=1$ and $b=1$.\\

For the following numerical simulations, we consider the damping factor $\lambda = 0.5$ and a time horizon of $T = 100$ units. This choice enables us to examine the long-term behavior of the system and to highlight the influence of the initial condition $\mathbf{x}_0$ on the resulting dynamics. As in the previous subsections, we begin by considering the strongly connected directed 3-cycle graph shown in Figure \ref{fig:numerical-oscillatory1}(a). To illustrate the behavior of the components of the solution $\mathbf{x}(t)=(\mathbf{x}_1(t),\mathbf{x}_2(t),\mathbf{x}_3(t))$ of the IVP in equation \eqref{eq:exponential-case}, we compute the corresponding numerical solution. In Figure \ref{fig:numerical-oscillatory1}(b) we consider the the initial condition $\mathbf{x}_0^\top=(0.25,0.5,0.25)$, while in Figure \ref{fig:numerical-expo1}(c) we take $\mathbf{x}_0^\top=(0.2,0.1,0.7)$.
\begin{figure}[H]
        \centering
        \includegraphics[width=1\linewidth]{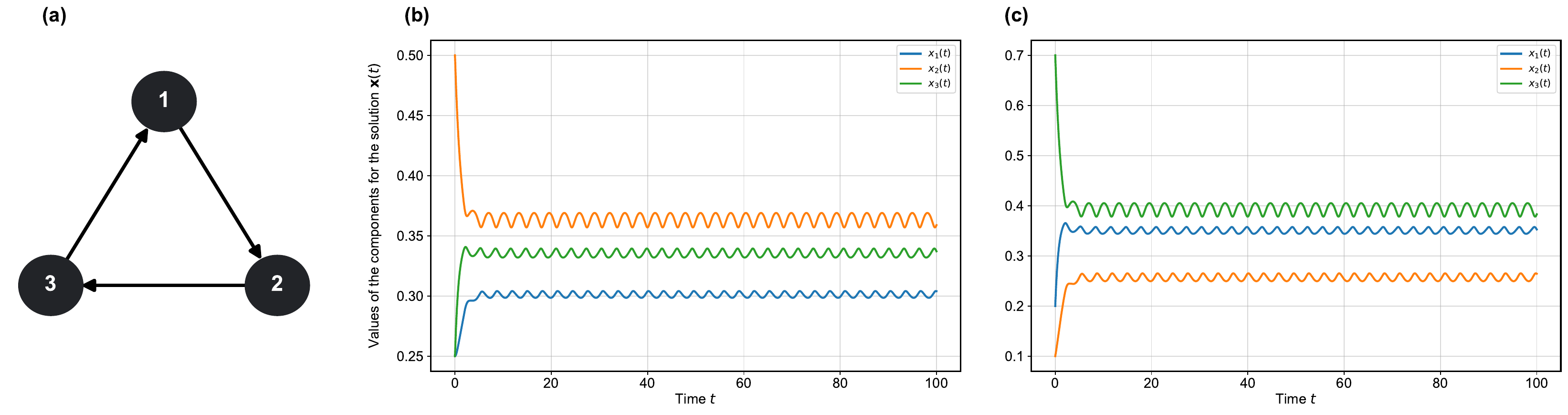}
        \caption{A numerical solution for the IVP in equation (\ref{eq:oscillatoy-case-numerical}) for the 3-cycle digraph and $\omega(t)=e^{t}\cos^2(t)$, for $t\geq0$.}
        \label{fig:numerical-oscillatory1}
    \end{figure}

Now, for the strongly connected digraph which has bidirectional edges between nodes $1$-$2$ and $2$–$3$ shown in Figure \ref{fig:numerical-oscillatory2}(a), we compute the numerical solution of the IVP in equation (\ref{eq:oscillatoy-case-numerical}). In Figure \ref{fig:numerical-oscillatory2}(b) we consider the initial condition $\mathbf{x}_0^\top=(0.3,0.2,0.5)$, while in Figure \ref{fig:numerical-oscillatory2}(c) we consider $\mathbf{x}_0^\top=(0.2,0.1,0.7)$.
\begin{figure}[H]
        \centering
        \includegraphics[width=1\linewidth]{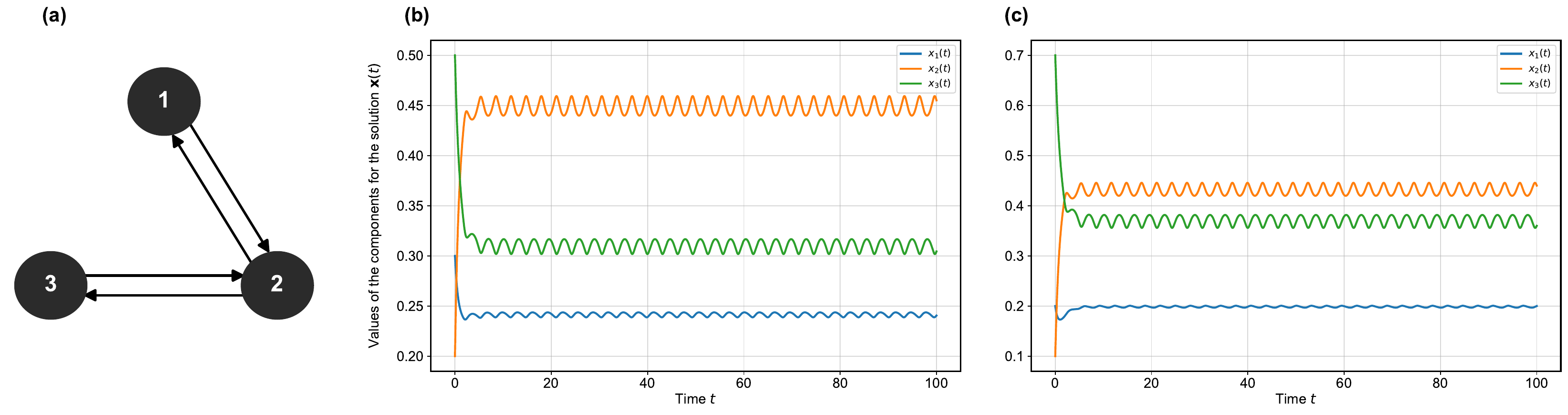}
        \caption{A numerical solution for the IVP in equation (\ref{eq:oscillatoy-case-numerical}) for the bidirectional digraph and $\omega(t)=e^{t}\cos^2(t)$, for $t\geq0$.}
        \label{fig:numerical-oscillatory2}
    \end{figure}

Additionally, for the non-strongly connected digraph in Figure \ref{fig:numerical-oscillatory3}(a), we report this oscillatory behavior by computing the numerical solution of the IVP in equation (\ref{eq:oscillatoy-case-numerical}). In Figure \ref{fig:numerical-oscillatory3}(b) we consider the initial condition $\mathbf{x}_0^\top=(0.5,0.3,0.2)$, while in Figure \ref{fig:numerical-oscillatory3}(c) we consider $\mathbf{x}_0^\top=(0.2,0.1,0.7)$.
\begin{figure}[H]
        \centering
        \includegraphics[width=1\linewidth]{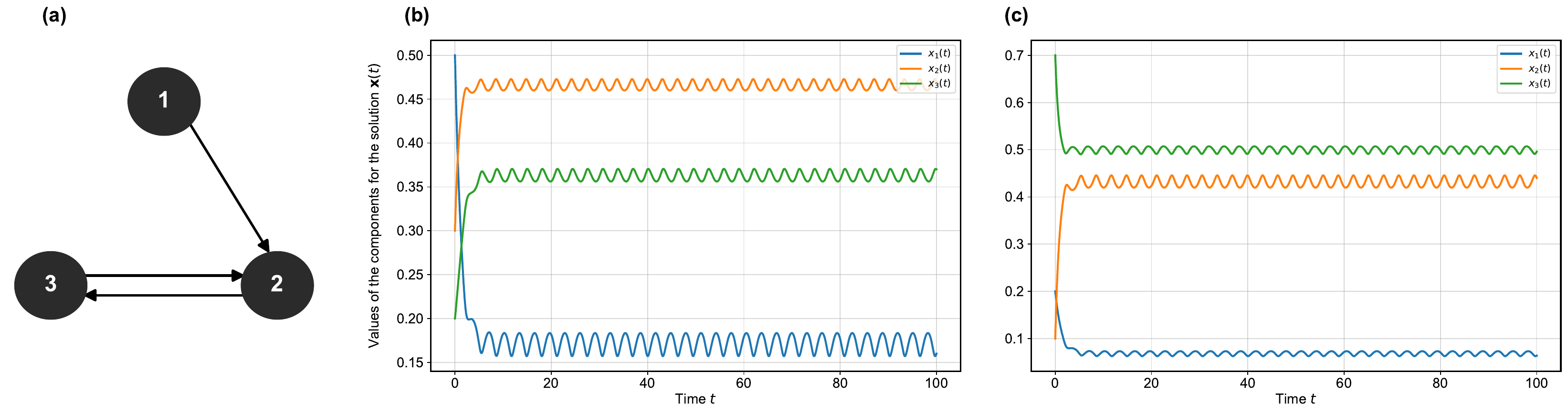}
        \caption{A numerical solution for the IVP in equation (\ref{eq:oscillatoy-case-numerical}) for a non-strongly connected digraph and $\omega(t)=e^{t}\cos^2(t)$, for $t\geq0$.}
        \label{fig:numerical-oscillatory3}
    \end{figure}
    
A straightforward observation is that, in all cases treated above, the components of the solution $\mathbf{x}(t)$ are asymptotically oscillatory and exhibit a strong dependence on the initial condition $\mathbf{x}_0 \in \mathbb{R}^{N \times 1}$. Moreover, in light of these numerical simulations, we conjecture that these oscillations are $\pi$-periodic. This statement is established in Theorem~\ref{thm:oscillatory-case}. Notably, unlike the previous subsection, 
certain components of the solution $\mathbf{x}(t)$ oscillate about values that differ from the corresponding components of the left Perron vector of the row-normalized matrix $P_A$.\\

Analogously to the exponential case in Subsection \ref{subsect:asymptotic-exponential}, the IVP in equation (\ref{eq:oscillatoy-case-numerical}) does not admit an explicit closed-form solution in the case of the exponential oscillatory memory function. In spite of this limitation  we still can study the asymptotic behavior of the solution using the results presented in \cite{Chicone}. In particular, we show the existence of asymptotically periodic solutions when the memory function $\omega(t)=e^{at}\cos^2(bt)$, for $t\geq0$, $(a,b>0)$ is considered. For the sake of completeness in the notation, we include the following result.\\

\begin{proposition}{\cite[Theorem 2.121, Section 2.4.3, Chapter 2]{Chicone}}\label{prop:Chicone-periodic}
    Let $\mathbf{A}(t)$ be a $T$-periodic matrix-valued function and let $\mathbf{b}(t)$ be a $T$-periodic row-matrix-valued function. If the differential equation
    \begin{equation*}
        \dot{\mathbf{x}}^\top(t)=\mathbf{x}^\top(t)\mathbf{A}(t)+\mathbf{b}(t)
    \end{equation*}
    has a bounded solution, then it has a $T$-periodic solution.
\end{proposition}

With this result at hand, we are in position to state the results on the existence of asymptotically $T$-periodic solutions for the IVP in equation \ref{eq:IVP-asymptotic} when we consider the oscillatory memory function mentioned above.\\

\begin{thm}\label{thm:oscillatory-case}
Let $A$ be a non-negative (not necessarily irreducible) square matrix of order $N$ of a directed graph $\mathcal{G}$ and let $P_A$ be its row-normalization described in Section \ref{Section:Notation}. Let us consider $a,b\in\R$ as positive parameters. Let $\omega(t)=e^{at}\cos^2(bt)$, for $t\geq0$, be the non-negative and local-integrable weight function.  Then, for any personalization vector 
$\mathbf{x}_0\in\R^{N\times1}$ with $\mathbf{x}_0\geq0$ and $\mathbf{x}_0^\top\mathbf{e}=1$ and for any $\lambda\in(0,1)$, the solution $\mathbf{x}$ of the Initial Value Problem (IVP) given by
\begin{equation}\label{eq:IVP-oscillatory}
\dot{\mathbf{x}}^\top(t)=\mathbf{x}^\top (t)(\lambda P_A-I_N)+(1-\lambda)\left[\frac{1}{\displaystyle\int_0^t \omega(u)\,du}\int_0^t \omega(t-u)\,\mathbf{x}^\top(u)\,du\right],\quad \mathbf{x}(0)=\mathbf{x}_0,
\end{equation}
is asymptotically periodic of period $\displaystyle\frac{\pi}{b}$.
\end{thm}

\begin{proof}
In this case, for the weight function $\omega(t)=e^{at}\cos^2(bt)$ for $t\geq0$, we have
\begin{equation*}
\int_0^t \omega(u)\,du=\int_0^t e^{au}\cos^2(bu)\,du= \frac{e^{at}}{2} \left[ \frac{1}{a} + \frac{a\cos(2bt) + 2b\sin(2bt)}{a^2 + 4b^2} \right] - \frac{1}{2}\left( \frac{1}{a} + \frac{a}{a^2 + 4b^2} \right),\quad (t\geq0),
\end{equation*}
where we have used the identity $\cos^2(bu)=(1+\cos(2bu))/2$. On the other hand, using the identity
\begin{equation*}
\cos^2(\,b(t-u)\,) = \frac{1}{2}\left[1+\cos(\,2b(t-u)\,)\right]=\frac{1}{2}\left[1+\cos(2bt)\cos(2bu) + \sin(2bt)\sin(2bu)\right],
\end{equation*}
we have
\begin{align*}\int_0^t \omega(t-u)\,\mathbf{x}^\top (u)\,du
&=\int_0^t e^{a(t-u)}\cos^2(\,b(t-u)\,)\mathbf{x}^\top(u)\,du\\
&= \frac{e^{at}}{2} \int_0^t e^{-au}\left[1+\cos(2bt)\cos(2bu) + \sin(2bt)\sin(2bu)\right]\,\mathbf{x}^\top(u)\,du\\
&= \frac{e^{at}}{2} \left[ \int_0^t e^{-au}\,\mathbf{x}^\top(u)\,du + \cos(2bt)\int_0^t e^{-au}\cos(2b u)\,\mathbf{x}^\top(u)\,du + \sin(2bt)\int_0^t e^{-au}\sin(2bu)\,\mathbf{x}^\top(u)\,du \right].
\end{align*}
At this point, we focus on the expressions depending on the parameters $a,b>0$
\begin{equation}\label{eq:improper-integrals}
 \mathbf{y}_1^\top(t) = \int_0^t e^{-au}\,\mathbf{x}^\top(u)\,du\quad,\quad \mathbf{y}_2^\top(t) = \int_0^t e^{-au}\cos(2bu)\,\mathbf{x}^\top(u)\,du\quad , \quad \mathbf{y}_3^\top(t) = \int_0^t e^{-au}\sin(2bu)\,\mathbf{x}^\top(u)\,du.   
\end{equation}
With these integral expression at hand, the IVP in equation (\ref{eq:IVP-oscillatory}) is rewritten as 

\begin{equation}\label{eq:asymptotic-periodic}
\dot{\mathbf{x}}^\top(t)=\mathbf{x}^\top(t)(\lambda P_A-I_N)+(1-\lambda)\,\mathbf{b}^\top(\,t,\mathbf{y}_1^\top(t),\mathbf{y}_2^\top(t),\mathbf{y}_3^\top(t)\,),\quad (t\geq0),
\end{equation}
where, after simplification of the factor $e^{at}/2$, we get
\begin{equation}\label{eq:almost-periodic}
\mathbf{b}^\top(\,t,\mathbf{y}_1^\top(t),\mathbf{y}_2^\top(t),\mathbf{y}_3^\top(t)\,)=\frac{ \mathbf{y}_1^\top(t) +\cos(2bt)\,\mathbf{y}_2^\top(t) +\sin(2bt)\,\mathbf{y}_3^\top(t) }{\displaystyle \left[ \frac{1}{a} + \frac{a\cos(2bt) + 2b\sin(2bt)}{a^2 + 4b^2} \right] - e^{-at}\left( \frac{1}{a} + \frac{a}{a^2 + 4b^2} \right) },\quad (t\geq0).
\end{equation}
Now, we study the asymptotic behavior of the equation (\ref{eq:asymptotic-periodic}). Firstly, as an application of Theorem \ref{thm:positivity} and Theorem \ref{thm:identity-to-one}, the solution $\mathbf{x}(t)$ of the IVP in equation (\ref{eq:IVP-oscillatory}) is a probability vector, that is, $\mathbf{x}(t)\geq0$ and $\mathbf{x}^\top(t)\,\mathbf{e}=1$ for all $t\geq0$. Moreover, for the vectors $\mathbf{y}_1^\top(t)$, $ \mathbf{y}_2^\top(t)$ and $\mathbf{y}_3^\top(t)$ in equation (\ref{eq:improper-integrals}) we have the boundedness
\begin{align*}\label{eq:limit-vector}
\Vert\,\mathbf{y}_1^\top(t)\,\Vert_1
&=\left\Vert \int_0^t e^{-au}\,\mathbf{x}^\top(u)\,du\, \right\Vert_1\leq \int_0^t e^{-au}\,\Vert\, \mathbf{x}^\top(u)\,\Vert_1\,du\leq \int_0^t e^{-au}\,du=\frac{1}{a}(1-e^{-at}),\\
\Vert\, \mathbf{y}_2^\top(t)\,\Vert_1
&=\left\Vert \int_0^t e^{-au}\,\cos(2bu)\,\mathbf{x}^\top(u)\,du \,\right\Vert_1\leq \int_0^t e^{-au}\,\vert \cos(2bu)\vert\,\Vert\, \mathbf{x}^\top(u)\,\Vert_1\,du\leq \int_0^t e^{-au}\,du=\frac{1}{a}(1-e^{-at}),\\
\Vert\, \mathbf{y}_3^\top(t)\,\Vert_1
&=\left\Vert\, \int_0^t e^{-au}\,\sin(2bu)\,\mathbf{x}^\top(u)\,du \,\right\Vert_1\leq \int_0^t e^{-au}\,\vert \sin(2bu)\vert\,\Vert\,\mathbf{x}^\top(u)\,\Vert_1\,du\leq\int_0^t e^{-au}\,du=\frac{1}{a}(1-e^{-at}).
\end{align*}
Therefore, since the absolute convergence of the improper integral implies the conditional convergence of the integral in the finite-dimensional space $\R^{N\times1}$, there exist constants vector $\mathbf{z}_1,\,\mathbf{z}_2,\,\mathbf{z}_3\in \R^{N\times1}$ (which depend on the parameters $a,b>0$) such that
\begin{align}
\lim_{t\to\infty} \mathbf{y}_1^\top(t) &= \int_0^\infty e^{-au}\,\mathbf{x}^\top(u)\,du=\mathbf{z}^\top_1,\nonumber\\
\lim_{t\to\infty} \mathbf{y}_2^\top(t) &= \int_0^\infty e^{-au}\cos(2bu)\,\mathbf{x}^\top(u)\,du=\mathbf{z}^\top_2,\nonumber\\
\lim_{t\to\infty}\mathbf{y}_3^\top(t) &= \int_0^\infty e^{-au}\sin(2bu)\,\mathbf{x}^\top(u)\,du=\mathbf{z}^\top_3.   
\end{align}
Now observe that, in the denominator of equation (\ref{eq:almost-periodic}), since $a>0$ we have $\displaystyle\lim_{t\to\infty}-e^{-at}\left( \frac{1}{a} + \frac{a}{a^2 + 4b^2} \right)=0$. Therefore, letting $t\to\infty$ and using the limit vectors from equation (\ref{eq:limit-vector}), we conclude that the function in equation (\ref{eq:almost-periodic}) converges a the continuous $\pi$-periodic function, denoted by$\mathbf{b}_{\infty}^\top(t)$, defined as
\begin{equation}\label{eq:pi-periodic-function}
\mathbf{b}_{\infty}^\top(t)=\frac{ \mathbf{z}_1^\top + \cos(2bt)\,\mathbf{z}_2^\top + \sin(2bt)\,\mathbf{z}_3^\top }{\displaystyle \frac{1}{a} + \frac{a\cos(2bt) + 2b\sin(2bt)}{a^2 + 4b^2} }.
\end{equation}  
Therefore, the original system in equation (\ref{eq:asymptotic-periodic}) approaches the limiting non-autonomous linear system
\begin{equation}\label{eq:time-periodic-system}
\dot{\mathbf{x}}^\top(t)=\mathbf{x}^\top(t)(\lambda P_A-I_N)+(1-\lambda)\,\mathbf{b}_{\infty}^\top(t),\quad (t\geq0).
\end{equation}
Observe that, since both the numerator and the denominator in equation (\ref{eq:pi-periodic-function}) are $\displaystyle\frac{\pi}{b}$-periodic matrix functions, the matrix function $\mathbf{b}_{\infty}^\top$ is also $\displaystyle\frac{\pi}{b}$-periodic.  Finally, since $(\lambda P_A-I_N)$ and $\mathbf{b}_{\infty}^\top(t)$ are $\displaystyle\frac{\pi}{b}$-periodic matrix functions, as an application of Proposition \ref{prop:Chicone-periodic}, the time-periodic system defined by equation (\ref{eq:time-periodic-system}) has a $\displaystyle\frac{\pi}{b}$-periodic solution.
\end{proof}

\section*{Acknowledgement}

This work has been supported by project M4066 (URJC Grant).

\end{document}